\documentclass[5p]{elsarticle}

\usepackage[intlimits]{amsmath}
\usepackage{amsthm,amssymb,amsfonts,times,mathptmx}
\usepackage{mathrsfs}
\usepackage[english]{babel}
\usepackage{xcolor}
\usepackage{tikz}
\usepackage{pgfplots}
\usepackage{placeins, booktabs, subcaption}

\pgfplotsset{compat=1.3}

\usepackage{graphics} 
\usepackage{epsfig}    

\theoremstyle{plain}

\newtheorem{lemma}{Lemma}
\newtheorem{theorem}{Theorem}

\theoremstyle{remark}

\biboptions{sort&compress}
\journal{Systems \& Control Letters}

\begin{document}
	
    \begin{frontmatter}
	\title{Stabilization of Linear Switched Systems with Long {Constant} Input Delay via Average or Averaging Predictor Feedbacks\tnoteref{t1}}
	\tnotetext[t1]{Funded by the European Union (ERC, C-NORA, 101088147). Views and opinions expressed are however those of the authors only and do not necessarily reflect those of the European Union or the European Research Council Executive Agency. Neither the European Union nor the granting authority can be held responsible for them.}
		
	\author[1]{Andreas Katsanikakis}
	\author[1]{Nikolaos Bekiaris-Liberis}
		
	\affiliation[1]{organization={Department of Electrical and Computer Engineering, Technical University of              Crete},
			addressline={University Campus, Akrotiri}, 
			city={Chania},
			postcode={73100}, 
			country={Greece}}
		
        \begin{abstract}
            We develop delay-compensating feedback laws for linear switched systems with time-dependent switching. Because the future values of the switching signal, which are needed for constructing an exact predictor-feedback law, may be unavailable at current time, the key design challenge is how to construct a proper predictor state. We resolve this challenge constructing two alternative, average predictor-based feedback laws. The first is viewed as a predictor-feedback law for a particular average system, properly modified to provide exact state predictions over a horizon that depends on a minimum dwell time of the switching signal (when it is available). The second is, essentially, a modification of an average of predictor feedbacks, each one corresponding to the fixed-mode predictor-feedback law. We establish that under the control laws introduced, the closed-loop systems are (uniformly) exponentially stable, provided that the differences among system’s matrices and among (nominal stabilizing) controller’s gains are sufficiently small, with a size that is inversely proportional to the delay length. Since no restriction is imposed on the delay, such a limitation is inherent to the problem considered (in which the future switching signal values are unavailable), and thus, it cannot be removed. The stability proof relies on multiple Lyapunov functionals constructed via backstepping and derivation of solutions’ estimates for quantifying the difference between average and exact predictor states. We present consistent numerical simulation results, which illustrate the necessity of employing the average predictor-based laws and demonstrate the performance improvement when the knowledge of a minimum dwell time is properly utilized for improving state prediction accuracy.
	\end{abstract}
		
		
    \end{frontmatter}

\section{Introduction}

\subsection{Motivation for Addressing Switched Systems with Input Delay}
Switched systems appear in numerous applications including traffic flow control \cite{traffic_flow}, automotive systems control \cite{IOANNOY}, \cite{ACC}, networked control systems \cite{telecom}, water networks \cite{Water}, power systems \cite{Luo2008}, \cite{PowerGrid}, and even epidemic spreading \cite{SIR_D}. Their dynamics may also be affected by the presence of input delays {that are large in magnitude,} which may degrade system's performance {or even render the system unstable when left uncompensated\footnote{The term ``long delay'' here refers to an input delay that may be so large that delay-free feedback laws (in particular, without a predictor structure), designed for the case in which there is no delay, can no longer guarantee satisfactory closed-loop performance or stability.}}. For example, in vehicle dynamics, switchings may occur due to changes between throttle/braking modes and input delays due to, e.g., engine response {and other vehicle's dynamics} \cite{IOANNOY}, \cite{ACC}, \cite{AICC}, \cite{Obs_P}, \cite{ISS_P}. {In particular, such large actuation delays in vehicular platoons with automated/connected
vehicles may lead to string instability (e.g., stop-and-go waves) when they are left uncompensated, see, e.g., \cite{DynamicsCV}, \cite{StringStab}.} In fault diagnosis in water distribution systems \cite{Water}, switchings may appear in the modeling of actuator dynamics, (e.g., due to transition of a pump between off/on states), while input delays may appear due to water transport times. {In power systems, inertia variations may be modeled as switchings \cite{PowerGrid}, while input delays may arise due to, for example, communication, where such communication/transmission delays may be much larger than the controller's sampling time \cite{FreqControlPowerSys}, \cite{VoltageAsync}. }Consequently, it is essential to design control strategies that can handle simultaneously switching and long input delay. 

\subsection{Literature Related to Linear Switched Systems with Input Delay}
Related literature includes results on stability analysis of switched systems with input delay utilizing Linear Matrix Inequalities (LMIs), \cite{FreqControlPowerSys}, \cite{Lin_discrete}, \cite{Xi_Corr}, \cite{Kp}, \cite{Switch_Communication}, \cite{B}, \cite{AHMED}, \cite{LIMA} and Lyapunov Krasovskii functionals, \cite{Yue_Kras}, \cite{Mazenc_LKF}, which, usually, impose a limitation on the delay length. Alternative methods that do not rely on an infinite-dimensional predictor-based controller, such as the Lyapunov based designs in \cite{Mazenc1}, \cite{D} or the truncated predictor methods in \cite{Wu_Truncated}, \cite{Sakthi_Truncated}, \cite{C}, also typically restrict the permissible size of the input delay or the system dynamics. For handling large input delays (for general systems) a predictor-based approach is required. However, the predictor-based approach in \cite{Lin_LTI} relies explicitly on the availability of future values of the switching signal, which may not be realistic in practical scenarios. Our results could also be viewed as related to works dealing with switched, stochastic or deterministic, hyperbolic PDE systems; see, for example, \cite{traffic_flow}, \cite{Auriol}, \cite{Prieur}. Although they address different problems, works dealing with asynchronous switchings, such as, e.g., \cite{A1}, \cite{ADT_asynchronous}, which may be viewed as employing delayed information of the switching signal in the control law, as well as with systems in which the switching signal is the manipulated variable that is subject to delay, such as, e.g., \cite{SwitchAffine}, \cite{E}, can be also viewed as relevant. Perhaps the most closely related existing results are the predictor-based techniques for linear systems (without switchings) with switched delays (stochastic or deterministic; see also, e.g., \cite{F}, \cite{G} that address switched delays), developed in \cite{Kong}, \cite{Kong_P1}, and \cite{A}. In fact, although the case of non-switched systems is considered in \cite{Kong}, \cite{Kong_P1}, the idea of constructing average predictor-based control laws for systems with stochastic (switched) input delays, introduced and analyzed utilizing the backstepping method in \cite{Kong}, has inspired our developments. 

\subsection{Contributions to Input Delay Compensation for Switched Systems}
In this paper, a switched linear system is considered with a long delay in the input. For the switching signal we consider that it is time dependent, featuring a minimum/maximum dwell time\footnote{However, our design approach can handle, as special case, the case of arbitrary/unknown switching; see \cite{my_ECC}, \cite{my_TDS}.}. Since an exact predictor-feedback law requires employment of the future switching signal values over a horizon equal to the delay, which may not be available (at current time), such an exact predictor-feedback law is inapplicable. For this reason, in this work, we propose two different delay-compensating predictor-based control laws as follows.

(i) In the first approach, we develop an average predictor-feedback control law. The key ideas of the design are construction of a predictor state on the basis of a properly chosen expected (average) system and proper utilization of the knowledge of a minimum dwell time of the switching signal (if available). In particular, in the case in which the switching signal's minimum dwell time is known, the controller can compute exact predictions of the state over a horizon that depends on the minimum dwell time and the last (before the current time) switching instant. (In the case of an unknown switching signal there may be no intervals of exact predictions.)

(ii) In the second approach, we introduce an alternative, predictor-based design that is, essentially, an average of exact predictor-feedback laws corresponding to the exact predictor feedbacks for each subsystem. Similarly to (i), when a minimum dwell time is available, the controller features intervals of exact prediction of the state, thus providing a more accurate prediction as compared with the case in which no switching signal information is available.

Since our feedback laws are not exact predictor-feedback laws, to establish (uniform) stability of the closed-loop systems we derive, in a constructive manner, estimates of the errors between our control laws and the exact predictor-feedback law. As the delay value may not be restricted, we have to necessarily impose a restriction on the magnitude of the differences among system’s matrices and among controller’s gains. This is an inherent limitation, in the practically realistic and theoretically challenging case in which the future values of the switching signal are not available (at current time), for guaranteeing that a sufficiently accurate (for stabilization) construction of the future values of the state is possible. We show however that there is a trade-off, between the allowable differences among system’s matrices and among controller’s gains, and the value of the delay. The stability analyses rely on introduction of suitable backstepping transformations, which enable construction of multiple Lyapunov functionals. In particular, we derive a lower bound on the allowable dwell time (that may depend proportionally to the delay value), which guarantees closed-loop stability. We provide a numerical simulation example that illustrates the effectiveness of both proposed controllers, including comparisons between the cases where dwell time information may or may not be available.

\subsection{Relation to Conference Papers \cite{my_ECC}, \cite{my_TDS}}
This work extends our results presented in conference papers \cite{my_ECC} and \cite{my_TDS} as follows. {Two different predictor-based feedback laws are developed via construction of an average \\ predictor-based law and via averaging exact predictor-feedback laws employing the dwell time knowledge, while \cite{my_ECC} and \cite{my_TDS} treat only the case when no dwell time information is available.} Thus, the new, average predictor-based control laws developed here may result in more accurate state prediction by properly incorporating minimum dwell time knowledge in the design; while they include as special case the designs from \cite{my_ECC}, \cite{my_TDS} when the switching signal is unknown and potentially arbitrary. {Accordingly, here we introduce a different stability analysis strategy in which we construct novel, multiple Lyapunov functionals (instead of a common Lyapunov functional as in \cite{my_ECC}, \cite{my_TDS} which imposes further restrictions on the system's parameters), while we introduce different backstepping transformations and derive improved bounds/estimates.} 

Furthermore, we present new numerical simulation results, implementing the new designs introduced here, which we also compare with the designs in \cite{my_ECC}, \cite{my_TDS}, illustrating that the former may considerably improve closed-loop performance. {In particular, we study in simulation the improvement in closed-loop performance with respect to the designs from \cite{my_ECC}, \cite{my_TDS} both quantitatively, via a certain performance index, and qualitatively, including performance comparison with respect to an ideal (though inapplicable), exact predictor-feedback law. We also quantify the robustness properties of our designs to small delay mismatch.}

\subsection{Organization of the Results Presented}
The outline of the paper is as follows. Section~\ref{sec2} presents the class of switched systems with input delay examined and defines the two different predictor-based control designs. In Section~\ref{sec3}, we state and prove our two main results, which are uniform exponential stability of the closed-loop systems under the proposed control laws. In Section~\ref{sec4} we provide consistent simulation results, including comparisons depending on the available information about the switching signal. Finally, in Section~\ref{sec5} we provide concluding remarks and discuss potential topics of future research.

\section{Problem Formulation and Control Design} \label{sec2}
\subsection{Switched Linear Systems with Input Delay}
We consider a linear switched system subject to a constant delay in the control input, described by the following dynamics
\begin{equation}\label{system}
\dot{X}(t) = A_{\sigma(t)} X(t) + B_{\sigma(t)} U(t-D), 
\end{equation}
where $X \in \mathbb{R}^q$ is the system state vector, $U \in \mathbb{R}$ is the control input, {with initial condition $U_0(s) = U(s)$, $s \in [-D,0]$, which is assumed to belong to $L^2([ -D,0]; \mathbb{R})$}, and $D>0$ is an arbitrarily long input delay. {Furthermore, the switching signal $\sigma:[0,+\infty) \rightarrow \mathcal{P}$, where $\mathcal{P} = \{ 1,2,\ldots, p \}$ is a finite set with $p$ being the number of different system modes, is a right-continuous, piecewise-constant function, assumed, according to the standard definition in \cite{Liberzon}, to have only finitely many discontinuities on any finite interval, thereby excluding Zeno behavior and ensuring well-posedness of Carathéodory solutions}\footnote{{
This follows using, e.g., \cite[Thm.~5.3]{HALE} in combination with \cite[Lem. 3]{Kong_P1}, from the linearity of the system dynamics (at each mode) and the assumption that $U_0\in L^2([-D,0];\mathbb{R})$.
}}. Denoting the strictly increasing sequence of switching times by $ \{0= t_0 < t_1 < t_2 < \ldots \}$, each interval $[t_k, t_{k+1})$ is a constant-mode interval. We consider switching signals that satisfy 
\begin{equation}\label{dwell_definition}
    \bar{\tau}_d \geq t_{k+1} - t_k \geq \tau_d, \quad \forall k \in \mathbb{Z}_{\geq0},
\end{equation} with $\bar{\tau}_d > 0$ and $ \tau_d > 0$ representing the maximum and the minimum time length between consecutive switchings, respectively. 

{We note that here we address the case in which the input delay $D$ is larger than the minimum dwell time $\tau_d$, i.e., $D > \tau_d$, as this setting represents the most challenging scenario because the system may switch to another mode before the delayed control action $U(t - D)$ takes effect at time $t$. Consequently, the future active mode over the prediction sub-horizon $[\tau_0(t)+\tau_d,\, t + D]$ (see Figure~\ref{fig1}) is unknown (in contrast to a case in which $D\leq \tau_d$), which prevents construction of an exact predictor-feedback law (for some $t$) and necessitates the construction of an average predictor-based design, as proposed here.  The results presented here are however still valid if the delay satisfies $D \geq \bar{\tau}_d>\tau_d$ as there is no restriction imposed on its size.}
In summary, the main challenge in our setup is to design a controller at time $t$ to predict the future state at $t+D$ employing only the information of the available current (or past), up to time $t$, switching signal values and the minimum dwell time.

\subsection{Predictor-Based Design via Average Predictor}

We propose the predictor-based controller for system (\ref{system}) as

\begin{equation}\label{U1}
    U(t)=U_1(t),
\end{equation}
with
\begin{align} \label{U1_exp}
 U_1(t) &=  \bar{K} \left( e^{\bar{A}\bigl(D-\tau(t)\bigr)} X\bigl(t+\tau(t)\bigr) \notag \right. \\ 
        & \, \left. + \int_{t+\tau(t)-D}^{t} {e^{\bar{A}(t-\theta)}\bar{B}U(\theta) \, d\theta} \right),
\end{align}
where
\begin{align}
\tau(t) &= \max \{0,  \tau_0(t) + \tau_d -t\} , \quad t \geq 0, \label{tr} \\
\tau_0(t) &= \max \{ t_k, \ k \in \mathbb{Z}_{\geq 0} \, | \, t_k \leq t   \}, \quad t \geq 0, \label{tr0}
\end{align}
and
\begin{align}\label{dwell_prediction}
    X\bigl(t+\tau(t)\bigr) &= e^{ A_{\sigma(\tau_0(t))} \tau(t)} X(t) \notag \\
                    & \ + \int_{t-D}^{t+\tau(t)-D} {e^{A_{\sigma(\tau_0(t))}(t+\tau(t)-D-\theta)}B_{\sigma(\tau_0(t))} U(\theta) \, d\theta}.
\end{align}
In particular, equation (\ref{dwell_prediction}) is the exact prediction of the state until time $t+\tau(t)$. Hence, the controller capitalizes on the availability of the minimum dwell time and the knowledge of the last (before the current time $t$) switching instant to construct a more accurate (in terms of being closer to the inapplicable exact predictor state) predictor state, as illustrated in Figure~\ref{fig1}.  The parameter $\tau(t)$ describes how far in the future the mode of the system remains known, and thus, the future state $X\bigl(t+\tau(t)\bigr)$ is precisely computed, as depicted in Figure~\ref{fig3}. {We note here that the less challenging case $D \leq \tau_d$ can be as well covered within our framework by slightly modifying (\ref{tr}) to $\tau(t) = \max \{0, \min \{\tau_0(t) + \tau_d -t,D\} \}, \,t \geq 0$ (for presentation and notational simplicity we do not belabor this case).} Notice that when $\tau(t) > 0$ the mode is known within the interval $[t, t+\tau(t)]$ due to the minimum dwell time assumption, allowing a precise short-term prediction. Beyond that interval, the uncertainty about the future mode, where the system will operate, necessitates the use of some expected matrices ($\bar{A}, \bar{B}$) chosen to approximate all modes in some (average) sense. For example, the matrices $\bar{A}$ and $\bar{B}$ can be viewed as expected values of the sets $\{A_1, A_2, \dots, A_p\}$ and $\{B_1, B_2, \dots, B_p\}$, respectively, and are chosen by the designer. The selection of $\bar{A}$ and $\bar{B}$ is critical, as stability depends on the norms $|A_i - \bar{A}|$ and $|B_i - \bar{B}|$, $i=1,\ldots,p$. 
 \begin{figure}[ht!]
    \centering
    \includegraphics[width=8 cm]{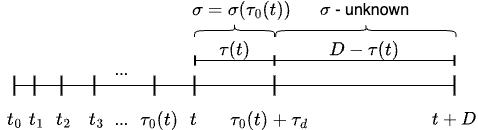}
    \caption{Switching instants and respective modes in the $[t, t+D]$ time interval.
    \label{fig1}}
\end{figure}
\begin{figure}[ht!]
    \centering
    \includegraphics[width=8 cm]{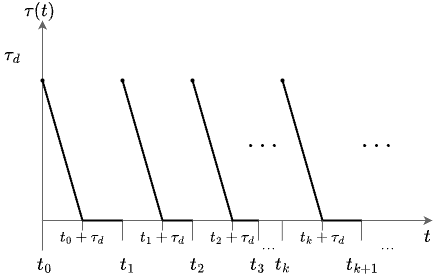}
    \caption{ Behavior of function $\tau(t)$ across several switching instants. Immediately after each switching $\tau(t)$ resets to $\tau_d$. It then decays linearly to zero until the next switching occurs (or remains at zero if no further switching occurs). 
    \label{fig3}}
\end{figure}

 {We now introduce an optimization routine to compute the expected matrices $(\bar{A}, \bar{B}, \bar{K})$ such that they minimize the deviations from the respective subsystems parameters (a procedure not included in the conference version \cite{my_TDS} that contains the respective preliminary control design).} Namely, one can solve the following optimization problem 
\begin{equation}\label{opt_routine}
    \min \, R \quad \text{subject to} \quad \left|\bar{Y} - Y_i\right| \leq R, \quad \forall i \in \mathcal{P},
\end{equation}
for a set of matrices $Y = \{Y_1, Y_2,\ldots,Y_p$\} of the same size, where \(| \cdot |\) denotes any desired, induced matrix norm. Since no further knowledge of the switching signal is assumed, which could help in computing a more accurate predictor state at time $t$, a computationally efficient and simple choice for $\bar{A}$ and $\bar{B}$ is to take them as the element-wise mean. Similarly, $\bar{K}$ represents a feedback gain determined by the user. In a similar manner one could solve the following minimization problem
\begin{equation}\label{opt_routine2}
    \min \, R \quad \text{subject to} \quad \left|\bar{K} - K_i\right| \leq R, \quad \forall i \in \mathcal{P},
\end{equation}
where each $K_i$ is designed such that $A_i + B_i K_i$ is Hurwitz for all $i = 1, \dots, p$. This choice is aligned with the choice of matrices $\bar{A}$, $\bar{B}$ used in the predictor-based control law. 

In fact, since we consider $D$ to be potentially long, there is no clear way to choose in the predictor-based control law (\ref{U1_exp}) a control gain that would potentially depend on the switching signal, as knowledge of $\sigma(t+D)$ is not available at time $t$ (because ${\tau_d} < D$). Therefore, due to the presence of input delay, even when (\ref{U1_exp}) corresponds to an exact predictor state, the choice of control gain may result, for mode $i$, in a nominal, closed-loop system in which matrix $A_i+B_iK_j$ may not be Hurwitz, as there is a mismatch between $K_j$ and a gain $K_i$ that makes $A_i+B_iK_i$ Hurwitz, where $i \neq j$, $i,j \in \mathcal{P}$. A similar problem is also reported in, e.g., \cite{Kp}, where the authors prove stability of the switched system, under certain conditions, which involve a restriction on the upper bound of delay value. In the presence of long input delay such problems arise typically and they originate in lack of synchronization between the active mode of the system and the gain applied see, e.g., \cite{A1}. Thus, we choose a control gain $\bar{K}$ that is independent of the switching rule. 

{Given the definition in equations (\ref{tr}), (\ref{tr0}), the case $\tau(t)=0$ appears when $t_k+\tau_d\leq t \leq t_{k+1}$, which is the interval over which the current mode cannot be guaranteed to remain constant. From that moment until the next switching instant, no future-mode information is available, and hence, $\tau(t)=0$. The case where $\tau(t) = 0, \ \forall t \geq 0$ can be used to describe the limiting case where the controller never utilizes dwell time knowledge, i.e., it operates without any knowledge about the switching signal.} This situation may appear in cases where the controller may have no access to the current or past switching instants, or is unaware of any information that concerns the dwell time. Under such conditions, the controller can still effectively operate without any short-term prediction capability, relying entirely on average approximations. This scenario can thus be interpreted as the switched system operating under arbitrary or unknown switching.
 
\subsection{Predictor-Based Design via Averaging Predictor Feedbacks}
We also design the following alternative controller
\begin{equation}\label{U2}
    U(t)=U_2(t),
\end{equation}
with
\begin{equation} \label{U2_exp}
 U_2(t) = \frac{1}{p} \sum_{i=1}^p K_i\hat{P_i}(t),  
\end{equation}
where $K_i$ is a feedback gain chosen such that $A_i + B_i K_i$ is Hurwitz and $\hat{P_i}(t)$ is the predictor state, for each subsystem, given by
 
\begin{equation}\label{dwell_subprediction}
\hat{P_i}(t) = e^{A_i \bigl(D-\tau(t)\bigr)} X\bigl(t+\tau(t)\bigr) + \int_{t+\tau(t)-D}^{t} {e^{A_i(t-s)}B_i U(\theta) \, d\theta},
\end{equation}
 for $X\bigl(t+\tau(t)\bigr)$ defined in (\ref{dwell_prediction}). Controller design (\ref{U2_exp}), (\ref{dwell_subprediction}) is based on calculating the exact predictor for each subsystem $i$ and then averaging over all these predictors (properly modified to account for the cases in which $\tau(t)>0$ for some times $t$). The motivation for such a design comes from the fact that it provides a simple control law that incorporates a delay-compensating mechanism, relying on averaging the exact predictors in the case in which the system would operate always in a single mode. However, as (\ref{U2}) does not correspond to an exact predictor, similarly as in the case of (\ref{U1}), one has to impose certain limitations on the system's matrices and controller's gains. On the other hand, the predictor-based control law (\ref{U1}) relies on constructing a predictor state for an average system (with matrices $\bar{A}$, $\bar{B}$) with input delay, while control law (\ref{U2}) relies on averaging the predictor states for each subsystem $i$. We further compare performance of the two alternative designs in simulation in Section \ref{sec4}.


\section{Stability Analysis}\label{sec3}  
\subsection{Stability Under (\ref{U1})}
\begin{theorem}\label{main_theorem}
Consider the closed-loop system consisting of plant (\ref{system}) and controller (\ref{U1}) under the assumption that the pairs $(A_i, \, B_i)$, for $i=1, \ldots,p$, are stabilizable. Let $K_i$ be such that $A_i+B_iK_i$, $i=1,...,p$, are Hurwitz, and thus, there exist some ${S_i= S_i^T} > 0 $, $ Q_i = Q_i^T > 0 $, satisfying 
    \begin{equation}\label{mean_delayf_stabilityy}
        \left({A_i} + {B_i} {K_i}\right)^T {S_i} +{S_i}\left({A_i} + {B_i} {K_i}\right) = -Q_i.
    \end{equation} There exist $\tau_d ^\star >0$, $\epsilon^\star > 0$, such that for any $\tau_d >\tau_d ^\star$, $\epsilon < \epsilon^\star$, where
    \begin{equation}\label{eps} 
        \epsilon= \max \limits_{i=1, \ldots, p} \{ |A_i-\bar{A}|, |B_i-\bar{B}|, |K_i-\bar{K}|\},
    \end{equation}
and for all $X_0\in\mathbb{R}^n$, $U_0\in L^2[-D,0]$, the closed-loop system is uniformly exponentially stable, in the sense that there exist  positive constants $\rho$, $\xi$, such that
    \begin{align}\label{tra}
        \left| X(t) \right| + & \sqrt{\int_{t-D}^{t} U(\theta)^2 d\theta} \leq \rho \left( |X(0)| + \sqrt{\int_{-D}^{0} U(\theta)^2 d\theta} \right) e^{-\xi t }, \notag \\
        & \, t \geq 0.
    \end{align} 
\end{theorem}
  Theorem \ref{main_theorem} does not explicitly impose a restriction on the delay value. However, it requires that the distance between any two different matrices in the sets $\{A_1,\ldots,A_p\}$ and $\{B_1,\ldots,B_p\}$ remains sufficiently small (see (\ref{eps_condition_2}) for an estimate of $\epsilon^*$), also depending on the delay value. This assumption is required for two main reasons, which are related to the choice of $\bar{K}$ and of $\left( \bar{A},\bar{B} \right)$ in (\ref{U1_exp}). The first is due to the mismatch between the pair $\left(A_i,B_i\right)$, for $i=1,\ldots,p$, of the actual future mode at which the system operates and the expected pair $\left(\bar{A},\bar{B}\right)$. Such a mismatch cannot be avoided by any predictor-based controller, because the system mode at the future time $t+D$ cannot be exactly predicted at current time $t$ (particularly under the assumption that the minimum dwell time is smaller than the delay, i.e., $\tau_d < D$, which is the hardest case). This also results in the requirement of restricting the maximum of $|A_i-\bar{A}|, |B_i-\bar {B}|$ as the possibility of the system operating always in a single mode, which is not known a priori, may not be excluded. Such a condition that all $\left|A_i - \bar{A} \right|$ and $\left|B_i - \bar{B} \right|$ are small, is analogous to the conditions in \cite{Kong_P1}, \cite{A} for the case of switching in delay values rather than in plant parameters. 
  The second is due to the choice of a fixed, average gain irrespectively of the different system's modes. The requirement that $\left|A_i - \bar{A} \right|$, $\left|B_i - \bar{B} \right|$, and $\left|K_i - \bar{K} \right|$ are small guarantees the existence of a lower bound for the minimum dwell time\footnote{%
  We note that such a bound would not be needed in the nominal case in which the delay-free closed-loop system is assumed to exhibit a common Lyapunov function.} 
  under the same $\bar{K}$ (see (\ref{dwell_conditionn}) for an estimate of $\tau_d^\star$), when employing multiple Lyapunov functionals to study stability. Furthermore, it is not clear whether a mode-independent choice of the nominal control gain can be avoided, since, due to the presence of input delay, any switching signal-dependent choice of a gain could lead to a potentially closed-loop system with delay-free, nominal dynamics dictated by a non-Hurwitz matrix of the form $A_i+B_iK_j$, where $i \neq j$, $i,j \in \mathcal{P}$. In fact, the proof of Theorem \ref{main_theorem}  shows that there exists a trade-off between the allowable distance among the system's matrices and the delay length. Note also that the minimum allowable dwell time $\tau_d^\star$, naturally depends on the delay $D$ (see Lemma \ref{trans_stability}). This dependence arises because the Lyapunov functionals and prediction-error bounds used in our analysis explicitly depend on the delay length. 
  
  {The stability conditions rely on a small-variations condition among the system’s matrices and among nominal control gains, which can be also translated as norm-bounded uncertainty around a nominal (expected) system. In this sense, our stability analysis can be viewed also as related to, for example, \cite{Mazenc1}, where exponential stability of switched systems with time-varying delays is studied under small system parameter variations. However, our work differs from \cite{Mazenc1} and related works adopting such a robust control framework in a) the control design, where here we develop {predictor-based control laws} that explicitly compensate {long, constant input delays}, without restricting the size of delays or assuming known future modes and b) the stability analysis where here we introduce a backstepping transformation that we utilize for construction of multiple Lyapunov functionals}.
  
  {Given that the stability conditions derived depend on (\ref{mean_delayf_stabilityy}) and the size of $\epsilon$ in (\ref{eps}), one can consider the following two optimization aspects in the control design stage. First, in the $U(t)=U_1(t)$ design (equations (\ref{U1}), (\ref{U1_exp})) through equations~(\ref{opt_routine}), (\ref{opt_routine2}), where the expected matrices $(\bar{A}, \bar{B}, \bar{K})$ are computed such as to {minimize the distance} from the respective subsystem matrices $(A_i, B_i, K_i)$. This ensures that the expected matrices employed in the predictor-based law (\ref{U1}), (\ref{U1_exp}) are constructed in an optimal way. Second, each local feedback gain $K_i$ can itself be obtained by an optimization-based or analytical/conventional design method for the corresponding {delay-free subsystem}, for example, by solving an LQR problem, or by using {pole-placement} techniques to set the eigenvalues of matrices $A_i + B_iK_i$. }

\subsection{Stability Under (\ref{U2})}

\begin{theorem}\label{main_theorem2}
Consider the closed-loop system consisting of plant (\ref{system}) and controller (\ref{U2}) under the assumption that the pairs \\ $(A_i, \, B_i)$, for $i=1, \ldots,p$, are stabilizable. Let $K_i$ be such that $A_i+B_iK_i$, $i=1,...,p$, are Hurwitz, and thus, there exist some ${S_i= S_i^T} > 0 $, $ Q_i = Q_i^T > 0 $, satisfying (\ref{mean_delayf_stabilityy}). There exist $\bar{\tau}_d ^\star >0$, $\bar{\epsilon}^\star> 0$, such that for any $\tau_d >\bar{\tau}_d ^\star$, $\bar{\epsilon} < \bar{\epsilon}^\star$, where
    \begin{equation}\label{eps2} 
        \bar{\epsilon}= \max \limits_{i,j=1, \ldots, p} \{ |A_i-A_j|, |B_i-B_j|, |K_i-K_j|\},
    \end{equation}
and for all $X_0\in\mathbb{R}^n$, $U_0\in L^2[-D,0]$, the closed-loop system is uniformly exponentially stable in the sense that there exist positive constants $\rho$, $\bar{\xi}$, such that
    \begin{align}\label{tra2}
        \left| X(t) \right| + &\sqrt{\int_{t-D}^{t} U(\theta)^2 d\theta} \leq\rho\left( |X(0)|  + \sqrt{\int_{-D}^{0} U(\theta)^2 d\theta} \right) e^{-\bar{\xi} t }, \notag \\ & \, t \geq 0.
    \end{align} 
\end{theorem}

As in Theorem \ref{main_theorem}, the restriction on $\bar{\epsilon}$ is imposed in order for (\ref{U2}) to be close to the exact, inapplicable predictor-feedback law and in order for the choice of control gains to not be far from the nominal (delay-free) stabilizing gains of each mode. One could potentially choose a single gain, as in (\ref{U1}), to obtain a simpler formula for the controller. Nevertheless, in both cases, the differences among the $A_i$, $B_i$, $K_i$ matrices have to still be restricted.
\subsection{ Proof of Theorem \ref{main_theorem}}
{The proof relies on a series of lemmas. First, Lemmas \ref{lemma exact predic}--\ref{inverse_transformation} introduce the backstepping transformation (and its inverse), via construction of the (inapplicable) exact predictor state, which maps the switched system (\ref{system}) with controller (\ref{U1}) to a suitable target system. Lemma~\ref{lemma_u(theta)} establishes norm equivalence between the state variables of the original and target system. Lemma~\ref{lemma_w(t)} then derives an explicit upper bound on the error of the average, predictor-based law, with respect to the ideal, exact predictor-feedback law, as a function of $\epsilon$.  
Using this bound, Lemma~\ref{trans_stability} establishes uniform exponential stability of the target system under sufficiently small $\epsilon$ and sufficiently large minimum dwell time $\tau_d$.
Finally, the proof of Theorem~\ref{main_theorem} is completed by establishing (\ref{tra}) through combining Lemmas~\ref{lemma_u(theta)}--\ref{trans_stability}.}

{\allowdisplaybreaks
\begin{lemma}(Exact predictor construction.)\label{lemma exact predic}
    Let the system (\ref{system}) experience $r$-switches,  within the interval $[t,t+D)$, $r \in \mathbb{N}_0$. Then the exact predictor $P(t)$ of this system is\footnote{%
    Throughout the paper, we define the product operator as 
    \[ \prod_{n=1}^{r+1} e^{A_{m_n}(s_n - s_{n-1})} \triangleq e^{A_{m_{r+1}}(s_{r+1}-s_r)} e^{A_{m_r}(s_r-s_{r-1})}\cdots e^{A_{m_1}(s_1 - s_0)}, \] i.e., the product is taken in descending order, from \(n=r+1\) to \(n=1\). No commutativity assumption on any matrices is made in the paper.}
        \begin{align}\label{P1(t)}
            P(t) &= \prod_{n=1}^{r+1} e^{A_{m_n}(s_n-s_{n-1})}X\bigl(t+\tau(t)\bigr) +  \sum_{n=1}^{r+1} \left( \prod_{j=n}^{r} e^{A_{m_{j+1}}(s_{j+1}-s_{j})} \right. \notag \\
               &\quad  \left. \times \int_{t+\tau(t)-D+s_{n-1}}^{t+\tau(t)-D+s_n}e^{A_{m_n}(t+\tau(t)-D+s_n-\theta)} B_{m_n} U(\theta) d\theta \right),
        \end{align}
where $\tau(t)$ is as in (\ref{tr}) and $X\bigl(t+\tau(t)\bigr)$ is as in (\ref{dwell_prediction}), $m_i \in \mathcal{P}$, for $i=1,2,\ldots,r+1$, denotes the mode of the system before the $i$-th (future) switching (after current time $t$), with $m_1=\sigma(\tau_0(t))$, and $s_i \in \mathbb{R}$, for $i=1,2,\ldots,r$, denotes the $i$-th switching instant, with $s_0=0$ and $s_{r+1}=D-\tau(t)$. 
\begin{figure}[ht!]
        \centering
        \includegraphics[width=8.7 cm]{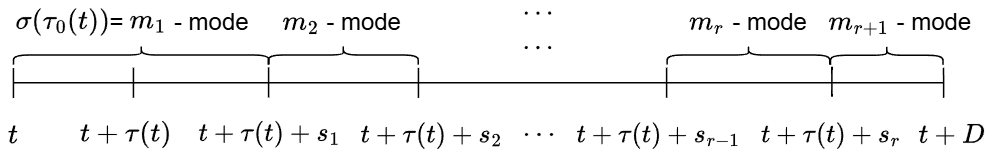}
        \caption{Future switching instants and respective modes in interval $[t, t+D]$.
        \label{fig2}}
        \end{figure}
    
    \begin{proof}
     {The predictor state is obtained by rewriting system (\ref{system}) as an ODE for $P(\theta)=X(\theta+D)$, then solving it on each constant–mode sub-interval, and subsequently applying the solution formula recursively to derive the explicit expression in \eqref{P1(t)}.} Towards this, setting $t=\theta + D$ and $P(\theta)=X(\theta+D)$, system (\ref{system}) becomes
        \begin{equation}\label{dP_theta_dwell}
            \frac{dP(\theta)}{d\theta}=A_{\sigma_{(\theta+D)}} P(\theta) + B_{\sigma_{(\theta+D)}} U(\theta).
        \end{equation}
        In the interval $[t, t+\tau(t)] \subseteq [t,t+D)$ the system operates in $\sigma(\tau_0(t))$ constant mode, as depicted in Figure \ref{fig2}, thus
        \begin{align}\label{P_theta_dwell}
            P(\theta) &= e^{{A_{\sigma(\tau_0(t))}}(\theta-t+D)} X(t) + \int_{t-D}^{\theta} e^{{A_{\sigma(\tau_0(t))}}(\theta - s)} {B_{\sigma(\tau_0(t))}}U(s) d s,
        \end{align}
        for $t-D \leq \theta \leq t+ \tau(t) - D  $. Setting $\theta=t + \tau(t) - D  $ we reach (\ref{dwell_prediction}). For  $ t + \tau(t) - D \leq \theta \leq t$  we further divide the interval $[t+\tau(t), t+D]$ into intervals of constant modes as
        \begin{align}\label{td_subintervals}
            t + \tau(t) - D + s_{i-1} &\leq \theta \leq t+ \tau(t) - D +s_{i},
        \end{align} for $i=1,2,\ldots,r+1$. In each sub-interval the system does not exhibit switching, and thus, we set ${m_i}=\sigma(\theta+D),$  for $i=1,...,r+1$. We can now proceed to the solution of (\ref{dP_theta_dwell}) {for each sub-interval for $\theta$ defined in (\ref{td_subintervals}),} which is extracted from the standard form of the general solution to a time-invariant ODE system as 
        \begin{align}\label{P_theta_dwell2}
            P(\theta) &= e^{{A_{m_i}}(\theta-t-\tau(t)+D-s_{i-1})} X(t+\tau(t)+s_{i-1}) \notag \\
                      &\quad + \int_{t+\tau(t)-D+s_{i-1}}^{\theta} e^{{A_{m_i}}(\theta - s)} {B_{m_i}}U(s) d s.
        \end{align} {
        Recalling that $P(\theta)=X(\theta+D)$, we expand recursively (\ref{P_theta_dwell2}). Hence for $i=1$ we have that
        \begin{align}\label{P_theta_dwellE4}
            X(t+\tau(t)+s_1) &= e^{{A_{m_1}}s_1} X(t+\tau(t)) \notag \\
            & \quad + \int_{t+\tau(t)-D}^{t+\tau(t)-D+s_{1}} e^{{A_{m_1}}(t+\tau(t)-D+s_{1} - s)} {B_{m_1}}U(s) d s,
        \end{align}
    for $X(t+\tau(t))$ given in (\ref{dwell_prediction}).
    Substituting (\ref{P_theta_dwellE4}) in (\ref{P_theta_dwell2}), for $i=2$ we obtain
        \begin{align}\label{P_theta_dwellE5}
            X(t+\tau(t)+s_2) &= e^{{A_{m_2}}(s_2-s_1)}e^{{A_{m_1}}s_1} X(t+\tau(t)) + e^{{A_{m_2}}(s_2-s_1)} \notag \\ &\quad \times \int_{t+\tau(t)-D}^{t+\tau(t)-D+s_{1}} e^{{A_{m_1}}(t+\tau(t)-D+s_{1} - s)} {B_{m_1}}U(s) d s \notag \\ 
            & \quad +\int_{t+\tau(t)-D+s_1}^{t+\tau(t)-D+s_{2}} e^{{A_{m_2}}(t+\tau(t)-D+s_{s} - s)} {B_{m_2}}U(s) d s.
        \end{align}
        Adopting the convention $\prod_{j=i}^{i-1}(\cdot)\equiv 1$, it follows, in a similar manner, from the recursive application of the solution formula (\ref{P_theta_dwell2}) and by induction that for any $i\in[1,r+1]$}
        \begin{align} \label{X(t+si)}
            X(t+\tau(t) + s_i) &= \prod_{n=1}^{i}{e^{{A_{m_n}}(s_n-s_{n-1})}} X\bigl(t+\tau(t)\bigr) \notag \\
                &\quad  + \sum_{n=1}^{i}\prod_{j=n}^{i-1}{e^{{A_{m_{j+1}}}(s_{j+1}-s_{j})}} \notag \\ &  \times \int_{t+\tau(t) - D + s_{n-1}}^{t+\tau(t) - D + s_n}e^{A_{m_n}(t+\tau(t)-D+s_n-\theta)}  B_{m_n}U(\theta)d\theta.
        \end{align}  
        Setting $i=r+1$ in (\ref{td_subintervals}) and $\theta=t$, from (\ref{P_theta_dwell2}) we reach
        \begin{align}\label{Pt_small}
            P(t)&=e^{A_{m_{r+1}}(D-\tau(t)-s_{r})} X(t+\tau(t)+s_{r}) \notag \\ 
                &\quad +\int_{t+\tau(t)-D+s_r}^{t} e^{A_{m_{r+1}}(t-\theta)} B_{m_{r+1}} U(\theta) d \theta . 
        \end{align} 
        Substituting (\ref{X(t+si)}) in (\ref{Pt_small}), we get (\ref{P1(t)}).    
    \end{proof}
\end{lemma}
}
\begin{lemma}(Backstepping transformation.)\label{Backstepping Transformation}
    The following backstepping transformation, 
        \begin{equation}\label{W_theta}
            W(\theta)=U(\theta) - K_{\sigma{(t+D)}} P(\theta), \quad t-D \leq \theta \leq t,
        \end{equation}
        where $P$ is obtained from (\ref{P_theta_dwell}) and (\ref{P_theta_dwell2}), transforms system (\ref{system}) to the target system 
        \begin{align}
            \dot{X}(t) &= \left(A_{\sigma(t)}+B_{\sigma(t)}K_{\sigma{(t)}}\right) X(t) + B_{\sigma(t)} W(t-D) \label{Xd_trans} , \\
                  W(t) &= U(t)-K_{\sigma{(t+D)}}P(t) ,\quad t\geq 0, \label{W_trans}
        \end{align}
    where $U$ and $P$ are given in (\ref{U1}) and (\ref{P1(t)}), respectively.

    \begin{proof}
    {Substituting the backstepping transformation into the plant dynamics (\ref{system}) and using $P(t-D)=X(t)$ shows that system \eqref{system} is mapped exactly to the target system \eqref{Xd_trans}, \eqref{W_trans}.} In more detail, system (\ref{system}) can be written as
        \begin{align}
            \dot{X}(t) &= \left(A_{\sigma(t)} + B_{\sigma(t)}K_{\sigma{(t)}}\right) X(t) \notag \\
                       &\quad + B_{\sigma(t)} \left(U(t-D) -  K_{\sigma{(t)}}X(t)\right). \label{2.16}
        \end{align}
        We now use (\ref{W_theta}). Setting $\theta=t-D$, from (\ref{P_theta_dwell}) we get $P(t-D)=X(t)$. Observing (\ref{P_theta_dwell}), (\ref{P_theta_dwell2}), and (\ref{2.16}), transformation (\ref{W_theta}) maps the closed-loop system consisting of the plant (\ref{system}) and the control law (\ref{U1}), to the target system (\ref{Xd_trans}), (\ref{W_trans}).
    \end{proof}
\end{lemma}

\begin{lemma}(Inverse backstepping transformation.)\label{inverse_transformation} 
    The inverse backstepping transformation of $W$ is
    \begin{equation}\label{inverse_theta}
        U(\theta)=W(\theta) + K_{\sigma{(t+D)}} \Pi(\theta),
    \end{equation}
    where 
    \begin{align}\label{Pi_theta_dwell}
        \Pi(\theta) &= e^{({A_{\sigma(\tau_0(t))}+B_{\sigma(\tau_0(t))}K_{\sigma(\tau_0(t))})}(\theta-t+D)} X(t) \notag \\
                    &\quad + \int_{t-D}^{\theta} e^{({A_{\sigma(\tau_0(t))}+B_{\sigma(\tau_0(t))}K_{\sigma(\tau_0(t))})}(\theta - s)} {B_{\sigma(\tau_0(t))}}W(s) d s,
    \end{align} for $t-D \leq \theta \leq  t+\tau(t)-D$, and
    \begin{align}\label{Pi_theta_dwell2}
        \Pi(\theta) &= e^{({A_{m_i}+B_{m_i}K_{m_i})}(\theta-t-\tau(t)+D-s_{i-1})} X(t+\tau(t)+s_{i-1}) \notag \\
                    &\quad + \int_{t+\tau(t)-D+s_{i-1}}^{\theta} e^{({A_{m_i}+B_{m_i}K_{m_i})}(\theta - s)} {B_{m_i}}W(s) d s,
    \end{align} for $\theta$ belonging to each interval defined in (\ref{td_subintervals}).

\begin{proof}
{The proof relies on analogous steps to the proof of Lem-ma~\ref{Backstepping Transformation} starting with the target system (\ref{Xd_trans}), based on which we obtain $\Pi$ in (\ref{Pi_theta_dwell}) by integration.} In more detail, we observe from (\ref{W_theta}) that  \[ U(\theta)=W(\theta) + K_{\sigma(t+D)}P(\theta).\] Solving the ODE (\ref{Xd_trans}) in a similar way as (\ref{dP_theta_dwell}) it can be shown that $\Pi(\theta)=X(\theta+D)$, where $\Pi(\theta)$ is given from (\ref{Pi_theta_dwell}) and (\ref{Pi_theta_dwell2}), and it holds that $\Pi(\theta)=P(\theta)$. 
\end{proof}
\end{lemma}
\begin{lemma}(Norm equivalency.)\label{lemma_u(theta)}  
    For the inverse transformation (\ref{inverse_theta})--(\ref{Pi_theta_dwell2}) the following inequality holds for some positive constant $\nu_1$ 
    \begin{equation}\label{ut_bound} 
      |X(t)|^2 + \int_{t-D}^{t}{| U(\theta) |^2 d \theta} \leq \nu_1 \left( |X(t)|^2 + \int_{t-D}^{t}{| W(\theta) |^2 d \theta} \right).
    \end{equation}
    Similarly, for the direct transformation (\ref{W_theta}) with (\ref{P_theta_dwell}), (\ref{P_theta_dwell2}),  it holds for some positive constant $\nu_2$
    \begin{equation}\label{wt_bound}
       |X(t)|^2 + \int_{t-D}^{t}{| W(\theta) |^2 d \theta} \leq \nu_2 \left( |X(t)|^2 + \int_{t-D}^{t}{| U(\theta) |^2 d \theta} \right).
    \end{equation}

    \begin{proof}
    {Using the direct and inverse backstepping transformations, combined with Young’s inequality (see, e.g., Appendix A of \cite{MirkinKrsticBook}) and derivation of bounds on the norms of the predictors (in terms of the norm of the respective systems' states), we derive estimates that establish norm equivalence between the norm of state $(X,U)$ and the norm of state $(X,W)$.} Towards that end, from (\ref{inverse_theta}), for the inverse transformation we apply Young's inequality to obtain 
        \begin{align}\label{U_theta_bound}
            \int_{t-D}^{t}{| U(\theta) |^2 d \theta} &\leq 2 \left( \ \int_{t-D}^{t}{ |W(\theta) |^2 d\theta} \right. \notag \\ 
                        &\left. \quad + {M_K}^2  \int_{t-D}^{t}{  \left| { \Pi(\theta)   } \right| ^2 d \theta}  \right),
        \end{align}
        {where for any matrix $R_i$, where $R$ can be $A,B,K, H,$ we set}
        \begin{equation}
         M_R = \max \{|R_1|,\ldots,|R_p|\}, \label{M_A}
        \end{equation}
        with $H_i=A_i + B_i K_i$, $i=1,\ldots,p$. In the interval $t-D \leq \theta \leq t+ \tau(t) - D $ from (\ref{Pi_theta_dwell}) we get
        \begin{align}\label{ut1_bound}
            &\int_{t-D}^{t+\tau(t)-D}{| U(\theta) |^2 d \theta} \leq 4 M_K^2 \tau(t) e^{2M_H \tau(t)}  |X(t)|^2  \notag \\
            & \quad \ \ \ \ + \left(2+4M_K^2 \tau(t)^2 e^{2 M_H \tau(t)} M_B^2 \right) \int_{t-D}^{t+\tau(t)-D}{| W(\theta) |^2 d \theta}.
        \end{align}
        For  $ t + \tau(t) - D \leq \theta \leq t$ divided in the sub-intervals as in (\ref{td_subintervals}), utilizing (\ref{P_theta_dwell2}) we get
        \begin{align}\label{ut2_bound}
            &\int_{t+\tau(t)-D}^{t}{| U(\theta) |^2 d \theta} \leq 4 M_K^2 \bigl(D-\tau(t)\bigr) e^{2M_H \bigl(D-\tau(t)\bigr)} \notag \\ 
            & \ \times |X\bigl(t+\tau(t)\bigr)|^2 + \left(2+4M_K^2 \bigl(D-\tau(t)\bigr)^2 e^{2 M_H \bigl(D-\tau(t)\bigr)} M_B^2 \right) \notag \\
            & \ + \left(2+4M_K^2 \bigl(D-\tau(t)\bigr)^2 e^{2 M_H \bigl(D-\tau(t)\bigr)} M_B^2 \right) \notag \\
            & \ \times \int_{t+\tau(t)-D}^{t}{| W(\theta) |^2 d \theta}.
        \end{align}
        Using triangle inequality and (\ref{dwell_prediction}), it holds that
        \begin{align}\label{X(t+tr)_bound_square}
            |X\bigl(t+\tau(t)\bigr)|^2 &\leq 2 e^{2 M_H \tau(t)} |X(t)|^2  \notag \\
            & \quad + 2 \tau(t) e^{2M_H \tau(t)} M_B^2 \int_{t-D}^{t+\tau(t)-D}|W(\theta)|^2 d\theta.
        \end{align}
        Combining (\ref{ut1_bound})--(\ref{X(t+tr)_bound_square}) we reach (\ref{ut_bound}),  where 
        \begin{align}\label{nuu1}
            \nu_{1} &=  \max \left\{ 4 M_{K}^2 D e^{2 M_H D}+1, 4 M_{K}^2 D^2 e^{2 M_H D} M_{B}^2 +2 \right\}.
        \end{align}   
        Analogously, using the direct transformation from (\ref{P_theta_dwell}), (\ref{P_theta_dwell2})  and (\ref{W_theta}), we can similarly prove (\ref{wt_bound}) via (\ref{P1(t)}), where
        \begin{align}\label{nuu2}
            \nu_{2} &=  \max \left\{ 4 M_{K}^2 D e^{2 M_A D}+1, 4 M_{K}^2 D^2 e^{2 M_A D} M_{B}^2 +2 \right\}.
        \end{align}
    \end{proof}
\end{lemma}

\begin{lemma}(Bound on error due to predictor mismatch.)\label{lemma_w(t)} Variable $W(t)$ defined in (\ref{W_trans}), under (\ref{U1}), satisfies
    \begin{equation}\label{Wt_bound}
        | W(t) | \leq \lambda(\epsilon,\tau(t)) \left(  |X(t)| +  \int_{t-D}^{t}{|U(\theta)|d\theta} \right), \quad t \geq 0,
    \end{equation}
    where $\lambda : \mathbb{R}_{+} \times [0,\tau_d] \to \mathbb{R}_{+}$ and for each fixed \(\tau \), the map \(\epsilon \mapsto \lambda(\epsilon,\tau)\) is a class \(K_{\infty}\) function of \(\epsilon\) defined in (\ref{eps}).
    
    \begin{proof}
    { The proof relies on the following steps. The error $W$, between the predictor-based law (\ref{U1}), (\ref{U1_exp}) and the exact, predictor-feedback law, is first rewritten in terms of the differences among system matrices and among control gains. Subsequently, via algebraic manipulations, we obtain an upper bound on $|W|$ in terms of the norm of the state $(X,U)$ with a factor $\lambda$ (as in (\ref{Wt_bound})), which we explicitly derive in terms of $\epsilon$ and $\tau(t)$. Towards this end, we proceed as follows.}
    
    For $W(t)$ defined in (\ref{W_trans}), using (\ref{P1(t)}) we write
        \begin{equation}\label{Wt}
            W(t) =  \Delta_1(t) + \Delta_2(t),
        \end{equation}
        where 
        \begin{align}
            \Delta_1(t) &=  \left(\bar{K} \prod_{n=1}^{r+1}{e^{{\bar{A}}(s_n-s_{n-1})}}- K_{\sigma{(t+D)}} \prod_{n=1}^{r+1}{e^{{A_{m_n}}(s_n-s_{n-1})}}\right) \notag \\
                        & \quad \times X\bigl(t+\tau(t)\bigr),  \label{D1t1}
        \end{align}
        \begin{align}
            \Delta_2(t) &=  \sum_{n=1}^{r+1} \left( \bar{K}  \prod_{j=n}^{r}{e^{\bar{A}(s_{j+1}-s_{j})}}\int_{t+\tau(t)-D+s_{n-1}}^{t+\tau(t)-D+s_n}e^{\bar{A}(t+\tau(t)-D+s_n-\theta)} \right. \notag \\ 
                        &\qquad \times \bar{B} U(\theta)d\theta  \notag \\ 
                        & - K_{\sigma{(t+D)}} \prod_{j=n}^{r}{e^{{A_{m_{j+1}}}(s_{j+1}-s_{j})}}\int_{t+\tau(t)-D+s_{n-1}}^{t+\tau(t)-D+s_n}e^{A_{m_n}(t-D+s_n-\theta)} \notag \\
                        &\qquad \left. \times B_{m_n}U(\theta)d\theta\right) \label{D2t1}.
        \end{align}
        For any matrix $\bar{R}, \, R_i$, we set         
        \begin{align}
            \Delta R_{i} &= \bar{R} - R_{i}, \label{Rmn}\\
            \epsilon_{R_{_{i}}} &= \left| \Delta R_{i} \right|, \label{epsR1}  \\
            \epsilon_R &= \max \limits_{i=1,\ldots,p} \{\epsilon_{R_{i}}\}, \label{epsR}      
        \end{align}    
        and since $m_{r+1}=\sigma(t+D)\in \mathcal{P}$, $t \geq 0$, we can re-write (\ref{D1t1}) as
        \begin{align}
            \Delta_1(t) &= \left( K_{\sigma{(t+D)}}\left( \prod_{n=1}^{r+1}{e^{{\bar{A}}(s_n-s_{n-1})}} -\prod_{n=1}^{r+1}{e^{{A_{m_n}}(s_n-s_{n-1})}}\right) \right.  \notag \\
                    &\left. \quad  + \Delta K_{m_{r+1}}\prod_{n=1}^{r+1}{e^{\bar{A}(s_n-s_{n-1})}} \right) X\bigl(t+\tau(t)\bigr). \label{D1t}
        \end{align}
        Using property $M'N'-MN=M'(N'- N) + (M'-M)N, $ where $M,N,M',N'$ denote arbitrary matrices, we write (\ref{D2t1}) as
        \begin{align}\label{d2t_new}
            \Delta_2(t) &=  \sum_{n=1}^{r+1} \left\{ K_{\sigma{(t+D)}} \left\{ \prod_{j=n}^{r} { e^{{\bar{A}}(s_{j+1}-s_{j})}}\left(Z_{1,n}(t) +Z_{2,n}(t)\right) \right. \right. \notag \\ & \quad \left. \left.  + Z_{3,n}(t)\right\} + \Delta_{K_{m_{r+1}}} \left(  \prod_{j=n}^{r}{e^{\bar{A}(s_{j+1}-s_{j})}} \right. \right. \notag \\
                        &\quad \left. \left.  \times \int_{t+\tau(t)-D+s_{n-1}}^{t+\tau(t)-D+s_n}e^{\bar{A}(t+\tau(t)-D+s_n-\theta)} \bar{B} U(\theta)d\theta \right)  \right\},
        \end{align}
        where
        \begin{align}
            Z_{1,n}(t) &=  \int_{t+\tau(t)-D+s_{n-1}}^{t+\tau(t)-D+s_n}{e^{\bar{A}(t+\tau(t)-D+s_n-\theta)}(\bar{B}-B_{m_n})U(\theta)d\theta},\label{Z1}  \\
            Z_{2,n}(t) &=  \int_{t+\tau(t)-D+s_{n-1}}^{t+\tau(t)-D+s_n}\left( e^{\bar{A}(t+\tau(t)-D+s_n-\theta)}\right. \notag \\ 
                       &\quad \left. - e^{A_{m_n}(t+\tau(t)-D+s_n-\theta)} \right) B_{m_n}U(\theta)d\theta, \label{Z2}
        \end{align}
        \begin{align}
            Z_{3,n}(t) &=  {\left(\prod_{j=n}^{r}{e^{\bar{A}(s_{j+1}-s_{j})}}-\prod_{j=n}^{r}{e^{{A_{m_{j+1}}(s_{j+1}-s_{j})}}}\right)} \notag \\ 
                       &\quad \times \int_{+\tau(t)-D+s_{n-1}}^{t+\tau(t)-D+s_n}e^{A_{m_{j+1}}(t+\tau(t)-D+s_n-\theta)}B_{m_n}U(\theta)d\theta  .  \label{Z3}
        \end{align}
    Setting $Y_1=A_{m_i}(s_i-s_{i-1}), \, Y_2=\Delta A_{m_i}(s_i-s_{i-1})$,
    for $m_i$, $s_i$ defined in Lemma \ref{lemma exact predic}, and using the fact that for any two $n \times n$ matrices $Y_1$, $Y_2$ the following inequality holds 
    \begin{align} \label{series}
       \left| e^{Y_1+Y_2}-e^{Y_1} \right| \leq |Y_2| e^{ |Y_1| }e^{ |Y_2| },
    \end{align} 
    where $|\cdot|$ denotes an arbitrary induced matrix norm\footnote{Proof of (\ref{series}) relies on the power series expansion for the matrix exponential and triangle inequality, see, e.g., \cite{Lie}.}, we have from (\ref{Rmn})--(\ref{epsR}), (\ref{series}), that
    \begin{align}\label{TheorResult}
        \left |e^{\bar{A}(s_n-s_{n-1})} - e^{A_{m_n} (s_n-s_{n-1})} \right | &\leq \epsilon_{A}  (s_n-s_{n-1}) e^{ \bar{M}_A (s_n-s_{n-1}) }   \notag \\ 
        & \quad  \times e^{ \epsilon_{A} (s_n-s_{n-1}) }  ,
    \end{align}  
    {where for any matrix $\hat{R}$, where $\hat{R}$ can be ${A}, {B},$ we set
    \begin{equation}
         \bar{M}_{\hat{R}} = \max \left\{\left|{\bar{\hat{R}}}\right|,M_{\hat{R}}\right\}. \label{M_ABAR}
    \end{equation}}
    We now upper bound the expression from (\ref{D1t}). We define
    \begin{equation}\label{Tk+1}
        T_{r+1}=\left|  \prod_{n=1}^{r+1}{e^{{\bar{A}}(s_n-s_{n-1})}-\prod_{n=1}^{r+1}{e^{{A_{m_n}}(s_n-s_{n-1})}}}  \right|.
    \end{equation}
    Developing (\ref{Tk+1}) for each iteration, for $r=0$, the result in (\ref{TheorResult}) can be directly applied to (\ref{Tk+1}). For $r=1$,
        \begin{align}
            T_2 = \left| e^{\bar{A}s_1}e^{\bar{A}(s_2-s_1)} -e^{A_{m_1}s_1} e^{A_{m_2}(s_2-s_1)}  \right|.
        \end{align}
        We expand the difference within the norm and using the triangle inequality for the norm bounds we get \begin{align}\label{n=2}
            T_2 & \leq \left| e^{\bar{A} s_1} \right| \left| e^{\bar{A}(s_2-s_1)} - e^{A_{m_2} (s_2-s_1)}   \right|  + \left| e^{A_{m_2} (s_2-s_1)} \right| T_1 .
        \end{align}
        We apply now (\ref{TheorResult}) to (\ref{n=2}), to obtain
        \begin{align}
            T_2 &\leq e^{\bar{M}_A  s_2}  \epsilon_A \left( e^{\epsilon_A s_2} (s_2-s_1) + e^{\epsilon_A s_1} (s_1-s_0)      \right) \notag \\
            &\leq e^{(\bar{M}_A  + \epsilon_A)s_2}  \epsilon_A  s_2.
        \end{align}
        For some $r$, we assume that the following expression holds
        \begin{align}\label{tkResult}
            T_{r} \leq e^{(\bar{M}_A  + \epsilon_A)s_r}  \epsilon_A  s_r.
        \end{align}
    We prove that the formula holds generally using the induction method.   
    Since (\ref{n=2}), (\ref{tkResult}) hold, expanding (\ref{Tk+1}) and using the triangle inequality we get
    \begin{align}\label{tk+1last}
        T_{r+1} & \le \left| e^{\bar{A}(s_{r+1}-s_r)}-e^{A_{m_{r+1}}(s_{r+1}-s_r)} \right|\left|\prod_{n=1}^{r}{e^{\bar{A}(s_n-s_{n-1})}}\right| \notag \\ 
               &\quad + \left|e^{A_{m_{r+1}}(s_{r+1}-s_r)}\right| T_r. 
    \end{align}
    Applying (\ref{TheorResult}) and (\ref{tkResult}) to (\ref{tk+1last}) we get
    \begin{align}\label{tk+1Result}
        T_{r+1} & \leq e^{(\bar{M}_A  + \epsilon_A)\bigl(D-\tau(t)\bigr)}  \epsilon_A  \bigl(D-\tau(t)\bigr),
    \end{align}
    which makes (\ref{tkResult}) legitimate for all $r$.
    Hence, applying any arbitrary induced matrix norm and substituting (\ref{tk+1Result}) in (\ref{D1t})  we get

    \begin{align}\label{d1t_bound}
        \left|\Delta_1(t)\right| &\leq e^{(\bar{M}_A  + \epsilon_A)\bigl(D-\tau(t)\bigr)} \left[ {M}_K \epsilon_A  \bigl(D-\tau(t)\bigr) + \epsilon_K   \right] \notag \\ 
        & \quad \times \left|X\bigl(t+\tau(t)\bigr) \right|.
    \end{align}
    Next, we observe that for any variable $\Psi$ that can be $U,W,\Pi$
        \begin{align}\label{series_trans}
            \int_{t+\tau(t)-D}^{t}{| \Psi(\theta) | d \theta} &= \sum_{i=1}^{r+1}{\int_{t+\tau(t)-D+s_{i-1}}^{t+\tau(t)-D+s_i}{| \Psi(\theta) | d \theta}}.
        \end{align}
     We can proceed now with bounding separately the terms in (\ref{d2t_new}), using the triangle inequality and the submultiplicative property. Applying any induced matrix norm in (\ref{Z1}), along with (\ref{epsR}), (\ref{series_trans}) we get
        \begin{align}\label{Z1bnd}
           \sum_{n=1}^{r+1}  \left | \prod_{j=n}^{r}{e^{{\bar{A}}(s_{j+1}-s_{j})}}  Z_{1,n}(t) \right| &\leq e^{ (\bar{M}_A + \epsilon_A) \bigl(D-\tau(t)\bigr)}   \epsilon_{B} \notag \\
           & \quad \times  \int_{t+\tau(t)-D}^{t}{|U(\theta)|d\theta}.
        \end{align}
        Similarly for the second term in (\ref{d2t_new}) we recall (\ref{TheorResult}), and hence,
        \begin{align}\label{Z2bnd}
            \sum_{n=1}^{r+1}  \left | \prod_{j=n}^{r}{e^{\bar{A}(s_{j+1}-s_{j})}}  Z_{2,n}(t) \right| &\leq 
            e^{ (\bar{M}_A + \epsilon_A) \bigl(D-\tau(t)\bigr)} \epsilon_A \bigl(D-\tau(t)\bigr) \notag \\
            & \quad \times \bar{M}_B \int_{t+\tau(t)-D}^{t}{|U(\theta)|d\theta}.
        \end{align}
        Recalling (\ref{tk+1Result}), we get similarly
        \begin{align}\label{Z3bnd}
            \sum_{n=1}^{r+1} {\left| Z_{3,n}(t) \right|} &\leq e^{ (\bar{M}_A + \epsilon_A) \bigl(D-\tau(t)\bigr)} \epsilon_A \bigl(D-\tau(t)\bigr) \notag \\ & \quad \times \bar{M}_B \int_{t+\tau(t)-D}^{t}{|U(\theta)|d\theta}.
        \end{align}
        Applying (\ref{epsR}), (\ref{series_trans})--(\ref{Z3bnd}), in (\ref{d2t_new}) we get 
        \begin{align}\label{d2t_Bound}
            \left| \Delta_2(t) \right| &\leq  e^{ (\bar{M}_A + \epsilon_A) \bigl(D-\tau(t)\bigr)} \left[ \left(\epsilon_B + 2\epsilon_A  \bigl(D-\tau(t)\bigr) \bar{M}_B\right)M_K \right. \notag \\
            & \quad \left. + \bar{M}_B  \epsilon_K  \right]\int_{t+\tau(t)-D}^{t}{|U(\theta)|d\theta}.
        \end{align}
        Next, we observe also that,
        \begin{align}\label{Xtr_bound}
            \left|X\bigl(t+\tau(t)\bigr)\right| &\leq e^{\bar{M}_A\tau(t)}\left|X(t)\right|  \notag \\ 
            & \quad +e^{\bar{M}_A\tau(t)}\bar{M}_B \int_{t-D}^{t+\tau(t)-D}{|U(\theta)|d\theta}.
        \end{align}
        Combining (\ref{d1t_bound}), (\ref{d2t_Bound}), and (\ref{Xtr_bound}) we reach (\ref{Wt_bound}), where
        \begin{equation}\label{error lambda}
            \lambda\left(\epsilon,\tau(t)\right)= \epsilon   e^{(\bar{M}_A + \epsilon)D} \max \{\delta_1(\tau(t)),\delta_2(\epsilon,\tau(t))\},
        \end{equation}
        and
        \begin{align}
            &\delta_1(\tau(t)) = \max\left\{1,\bar{M}_B\right\}  \left[M_K \bigl(D-\tau(t)\bigr) + 1   \right],  \label{error delta1} \\
            &\delta_2(\epsilon,\tau(t))=   \left[2M_K \bar{M}_B \bigl(D-\tau(t)\bigr) + \epsilon + M_K + \bar{M}_B  \right]. \label{error delta2}
        \end{align} 
    \end{proof}
\end{lemma}

\begin{lemma}(Stability of target system.)\label{trans_stability}
Assume that $(A_i,B_i)$, $i=1,...,p$, are stabilizable pairs and let $K_i$ be such that $A_i+B_iK_i$, $i=1,...,p$, are Hurwitz, and thus, there exist some ${S_i= S_i^T} > 0 $, $ Q_i = Q_i^T > 0 $, satisfying (\ref{mean_delayf_stabilityy}). For any $\epsilon < \epsilon^\star$, where $\epsilon$ is defined in (\ref{eps}) and
    \begin{align}\label{eps_condition_2}
        \epsilon^\star &= \min \left\{ \bar{\lambda}^{-1}\left(\frac{1}{ \sqrt{2 e^D D \nu_{1}}}\right), \notag \right. \\
        &\quad \left. \bar{\lambda}^{-1}\left(\frac{\min \limits_{i=1,\ldots,p} \left\{ \frac{\lambda_{\min}(Q_i)}{|B_i {S_i}|} \right\}}{ \sqrt{2 e^D D (\nu_{1}+1})} \right)  \right\},
    \end{align}  
    where $\bar{\lambda}^{-1}$ denotes the inverse function of $\bar{\lambda}(\epsilon)=\lambda(\epsilon,0)$, for $\lambda$ defined in Lemma \ref{lemma_w(t)}, there exist positive constants $\mu$, $\alpha$, $\beta$, such that the target system (\ref{Xd_trans}), (\ref{W_trans}) with (\ref{U1}), (\ref{U1_exp}), and (\ref{P1(t)}) is uniformly exponentially stable for every switching signal $\sigma$ satisfying (2) with $\tau_d >\tau_d^\star$, where     
    \begin{equation}\label{dwell_conditionn}
        \tau_d^\star= \frac{\ln \mu}{\beta}.
    \end{equation}
    In particular, the following holds
    \begin{align}\label{trans_stability_result1}
        \left| X(t) \right|^2 + & \int_{t-D}^{t} W(\theta)^2 d\theta \leq \mu \left( |X(0)|^2 + \int_{-D}^{0} W(\theta)^2 d\theta \right) \notag \\ 
        & \quad \quad  \quad \times   e^{-\left(\beta- \frac{\ln \mu}{\tau_d} + \frac{\alpha}{\bar{\tau}_d} \right)t} , \quad t \geq 0.
    \end{align}

\begin{proof}
 {The proof relies on the following steps. A Lyapunov functional is first constructed for each mode of the target system (\ref{Xd_trans}), (\ref{W_trans}) under (\ref{U1}), (\ref{U1_exp}), and its derivative is subsequently shown to be negative definite employing estimate (\ref{Wt_bound}) for $W$ from Lemma~\ref{lemma_w(t)}. Finally, the dwell time condition ensures exponential decay across switchings. Towards this end, we proceed as follows.}

    According to Lemma \ref{Backstepping Transformation}, for each subsystem of the family described by the target switched system (\ref{Xd_trans}), (\ref{W_trans}) it holds
    \begin{align}
            \dot{X}(t) &= \left(A_{i}+B_{i}K_{i}\right) X(t) + B_{i} W(t-D) \label{Xdp_trans} , \\
                  W(t) &= U(t)-K_{g}P(t), \label{Wp_trans}
        \end{align}  $ \forall i, \,g \in \mathcal{P}$.
        Consider $[t_k, t_{k+1})$ a time window where the (\ref{Xdp_trans}), (\ref{Wp_trans}) system is operating without switchings. We adopt now the following Lyapunov functional 
    \begin{equation}\label{LyapunovFnc}
        V_i(t) = X(t)^T{S_i}X(t) + b_i \int_{t-D}^{t} e^{(\theta+D-t)}W(\theta)^2 d\theta.
    \end{equation}
     Calculating the derivative of (\ref{LyapunovFnc}), along the solutions of the target subsystem for $t \in [t_k , t_{k+1})$, we obtain
    \begin{align}
        \dot{V}_{{i}}(t) &\leq -X(t)^T  Q_{i}  X(t) + B_{i}^T W(t-D){S_i} X(t) \notag \\ \
        &\quad  + X(t)^T {S_i} B_{i} W(t-D) + b_{i} e^D W(t)^2 \notag \\
        &\quad - b_{i} W(t-D)^2 - b_{i}\int_{t-D}^{t} e^{(\theta+D-t)}W(\theta)^2 d\theta. \label{dv_t_1}
    \end{align}
    Observing $-X(t)^T Q_{i}  X(t) \leq - \lambda_{\min}({Q_{i}}) |X(t)|^2$, applying \\ Young's inequality (see, e.g., Appendix A in \cite{MirkinKrsticBook}), and choosing
    \begin{equation}
        b_{i} = \frac{ 2 |B_{i}{S_i}|^2 }{ \lambda_{\min}({Q_{i}}) }, 
    \end{equation}
    we get from (\ref{dv_t_1}) that
    \begin{align}
    \label{dv_t_11}
        \dot{V}_{i}(t) &\leq -\frac{1}{2} \lambda_{\min}({Q_{i}}) |X(t)|^2 \notag \\
        & \quad + b_{i}e^D W(t)^2 - b_{i}\int_{t-D}^{t} e^{(\theta+D-t)}W(\theta)^2 d\theta.
    \end{align}
    Using Lemmas \ref{lemma_u(theta)} and \ref{lemma_w(t)} we get 
    \begin{align}\label{wt_sqrt11}
        W(t)^2 &\leq 2  \lambda^2(\epsilon,\tau(t))(D \nu_1 + 1) |X(t)|^2 + 2 \lambda^2(\epsilon,\tau(t)) D \nu_1 \notag \\
        &\qquad \times \int_{t-D}^{t}e^{(\theta+D-t)}{W(\theta)^2 d\theta}.
    \end{align}    
    Employing (\ref{wt_sqrt11}) in (\ref{dv_t_11})  we get
    \begin{align}\label{dv_t_12}
        \dot{V}_{i}(t) &\leq - \left( \frac{1}{2} \lambda_{\min}({Q_{i}}) - 2 b_{i} e^D \lambda^2(\epsilon,\tau(t))(D \nu_1 + 1) \right) |X(t)|^2   \notag \\
                   &\quad -b_{i} \int_{t-D}^{t}{ \left[1 - 2 e^D  \lambda^2(\epsilon,\tau(t)) D \nu_1        \right] e^{(\theta+D-t)} W(\theta)^2 d\theta } . 
    \end{align}
    In order to preserve negativity in (\ref{dv_t_12}) we require
    \begin{align}
        \frac{1}{2} \lambda_{\min}({Q_{i}}) - 2 b_{i} e^D \lambda^2(\epsilon,\tau(t))(D \nu_1 + 1) &> 0,  \label{a1p} \\
        1 - 2 e^D  \lambda^2(\epsilon,\tau(t)) D \nu_1 &> 0 ,  \label{a2p}
    \end{align}
    $t\in [t_k,t_{k+1})$. For fixed $\epsilon$, we observe from (\ref{error lambda})--(\ref{error delta2}) that the left-hand sides of (\ref{a1p}), (\ref{a2p}) are increasing functions in $\tau$. Thus, since $\tau : \mathbb{R_+} \rightarrow [0,\tau_d]$, then negativity is preserved under
    \begin{align}
            \frac{1}{2} \lambda_{\min}({Q_{i}}) - 2 b_{i} e^D \lambda^2(\epsilon,0)(D \nu_{1} + 1) &> 0, \label{a11p}  \\
            1 - 2 e^D  \lambda^2(\epsilon,0) D \nu_{1} &> 0 , \label{a22p}
    \end{align} 
    which hold under the restriction on $\epsilon$ in the statement of the lemma. Therefore, we conclude that 
    \begin{equation}\label{final derivative of V}
        \dot{V}_{i}(t) \leq {-a_{i,k}(t) }V_{i}(t),\quad t \in [t_k, t_{k+1}),
    \end{equation}
    where 
    \begin{align}\label{alpha}
        a_{i,k}(t) &= \min \left \{ 1-2e^D \lambda^2(\epsilon,\tau(t)) D \nu_{1}, \right. \notag \\
        &\left. \quad \frac{  \frac{1}{2}\lambda_{\min}(Q_{i}) - 2 b_{i}e^D \lambda^2(\epsilon,\tau(t))(D \nu_{1} + 1) }{\lambda_{\max}({S_i})} \right \}.
    \end{align}
    Functions $a_{i,k}(t)$ are strictly positive and decreasing for \\ $t\in [t_k,t_k+\tau_d]$, while for $ t_k+\tau_d \leq t < t_{k+1}$, we have \(a_{i,k}(t) = a_{i,k}(t_{k}+\tau_{d}) \). Moreover, we set
    \begin{equation}\label{betaa}
        \beta_i = a_{i,k}(t_{k}+\tau_{d}),
    \end{equation}
    i.e.,
        \begin{align}\label{alpha1}
            \beta_i &= \min \left \{ 1-2e^D \lambda^2(\epsilon,0) D \nu_{1}, \right. \notag \\
                    &\left. \quad \frac{  \frac{1}{2}\lambda_{\min}(Q_{i}) - 2 b_{i}e^D \lambda^2(\epsilon,0)(D \nu_{1} + 1) }{\lambda_{\max}({S_i})} \right \},
        \end{align}
        which also implies that
        \begin{align}
                a_{i,k}(t) &> \beta_i, \quad  t_k \leq t < t_k + \tau_d, \label{inequality for a=b} \\
                a_{i,k}(t) &= \beta_i, \quad  t_k+\tau_d \leq t < t_{k+1}.\label{inequality for a<b}
        \end{align}
        From (\ref{final derivative of V}) along with (\ref{inequality for a=b}), (\ref{inequality for a<b}) we reach 
        \begin{align}\label{derivative_subsystems}
            \dot{V}_{i}(t) \leq 
            \begin{cases} 
                -a_{i,k}(t) V_{i}(t), & t_k \leq t < t_k + \tau_d \\
                -\beta_i V_{i}(t), & t_k + \tau_d \leq t < t_{k+1}
            \end{cases}.
        \end{align}
        Using the comparison principle and integrating (\ref{derivative_subsystems}) we get
            \begin{equation}\label{vi_a}
                V_{i}(t) \leq e^{- \int_{t_k}^t a_{i,k}(u)du} V_{i}(t_k), \quad t_k \leq t < t_k + \tau_d.
            \end{equation}
        For \(t_k + \tau_d \leq t < t_{k+1}\), integrating (\ref{derivative_subsystems}) and using (\ref{inequality for a=b}) we get
        \begin{equation}\label{vi_b}
            V_{i}(t) \leq e^{-\beta_i(t - (t_k + \tau_d))} V_{i}(t_k + \tau_d), \quad t_k + \tau_d \leq t < t_{k+1}. 
        \end{equation}
        Hence, from (\ref{vi_a}) for $t=t_k + \tau_d$, we arrive at
        \begin{equation}\label{vi_tk+td}
            V_{i}(t_k + \tau_d) \leq e^{- \int_{t_k}^{t_k + \tau_d}a_{i,k}(u)du} V_{i}(t_k) .
        \end{equation}      
        Substituting (\ref{vi_tk+td}) into (\ref{vi_b}), we arrive at
        \begin{equation}{\label{inequality for all lyaps}}
            V_{i}(t) \leq
            \begin{cases}
                e^{- \int_{t_k}^ta_{i,k}(u)du} V_{i}(t_k), \quad t_k \leq t < t_k + \tau_d \\
                e^{-\int_{t_k}^{t_k+\tau_d} {\left(a_{i,k}(u)- \beta_i\right) du} - \beta_i (t - t_k)} V_{i}(t_k), \\ \quad  t_k + \tau_d \leq t < t_{k+1}
            \end{cases}.
        \end{equation}
        We define next
        \begin{align}
         \beta &= \min \limits_{i=1,\ldots,p} \{\beta_i\}, \label{bita}
        \end{align}
        and
        \begin{equation}\label{alpa}
            \alpha_k =  \min \limits_{i=1,\ldots,p} \left\{ \int_{t_k}^{t_k+\tau_d} {\left( a_{i,k}(u)- \beta \right) du} \right\}, \ k \in \mathbb{Z}_{\geq0}.
        \end{equation}
        We can further re-write (\ref{inequality for all lyaps}), with (\ref{bita}), (\ref{alpa}), as,
        \begin{equation}{\label{inequality for all lyaps22}}
            V_{i}(t) \leq e^{-\Gamma_k(t)} V_{i}(t_k), \quad  t \in [t_k,t_{k+1}),        
        \end{equation}
        where
        \begin{equation}\label{gamma}
            \Gamma_k(t) = 
            \begin{cases}
                \int_{t_k}^t a_{i,k}(u)du, \quad t_k \leq t < t_k + \tau_d \\
                \alpha_k + \beta (t - t_k),  \quad t_k + \tau_d \leq t < t_{k+1}
            \end{cases}.
        \end{equation}
        Result (\ref{inequality for all lyaps22}), together with (\ref{inequality for a=b}), (\ref{gamma}), show stability of individual, target subsystems (with decay rate at least $\beta$). We show next stability of the target, switched system. Pick an arbitrary \( t > 0 \), let \( t_0 := 0 \), and denote the switching times on the interval \( (0,t) \) by $\left\{ t_1,t_2, \dots,t_{l-1}, t_{l} \right\}$, where $l=N_{\sigma}(0,t)$, for $N_{\sigma}(0,t)$ representing the number of total switchings over the interval $[0,t)$. From (\ref{LyapunovFnc}) we have
        \begin{align}\label{kappas}
        \kappa_{1,i}&\left(\left| X(t) \right|^2 + \int_{t-D}^{t} W(\theta)^2 d\theta \right)  \leq V_i(t) \notag \\
                &\leq  \kappa_{2,i}\left(\left| X(t) \right|^2 + \int_{t-D}^{t} W(\theta)^2 d\theta \right),
    \end{align}
    where 
    \begin{align}
        \kappa_{1,i} &= \min \left\{\lambda_{\min}({S_i}), \frac{2|{S_i}B_i|^2}{\lambda_{\min}(Q_i)}\right\}, \\
        \kappa_{2,i} &= \max \left\{\lambda_{\max}({S_i}), \frac{2|{S_i}B_i|^2}{\lambda_{\min}(Q_i)}e^D\right\}, 
    \end{align}
    for any $i \in \mathcal{P}$.
    Setting \begin{equation} \label{const_mu}
        \mu = \frac{\kappa_{2}}{\kappa_{1}},
    \end{equation} where 
    \begin{align}
        \kappa_1 &= \min \limits_{i=1,\ldots,p} \{ {\kappa_{1,i}}\}, \label{k1} \\
        \kappa_2 &= \max \limits_{i=1,\ldots,p}\{\kappa_{2,i}\}, \label{k2}
    \end{align} and from (\ref{kappas}), (\ref{k1}), (\ref{k2}), $\forall \sigma(t_{k}) = i, \sigma(t_{k}^-) = j, \\ \forall i,j \in \mathcal{P}, \ i \neq j$, and $ \forall t_k \in [t_1,t_2, \dots,t_{l-1}, t_{l}]$, we get
    \begin{equation}\label{switch_mu}
        V_i({t_{k}}) \leq \mu V_j({t_{k}}).
    \end{equation} 
    For $t\geq t_l$, we apply (\ref{inequality for all lyaps22}) and we get 
    \begin{equation}\label{switch}
        V_{\sigma(t)}(t) \leq e^{-\Gamma_l(t) } V_{\sigma(t_l)}(t_l).
    \end{equation} 
    Applying (\ref{switch_mu}) in (\ref{switch}) we arrive at
    \begin{equation}\label{switch1}
        V_{\sigma(t)}(t) \leq  e^{-\Gamma_l(t) } \mu V_{\sigma(t_{l}^-)}(t_{l}).
    \end{equation}  
     For the window $t \in [t_{l-1}, t_{l})$ we apply again (\ref{inequality for all lyaps22}), setting $k=l-1$ 
    \begin{align}\label{previous_switch}
        V_{\sigma(t_{l}^-)}(t_{l}) &\leq  e^{-\Gamma_{l-1}(t_{l}^-) } V_{\sigma(t_{l-1})}(t_{l-1}) \notag  \\ 
        &=e^{-( \alpha_{l-1} + \beta (t_l - t_{l-1}))} V_{\sigma(t_{l-1})}(t_{l-1}).
    \end{align}
    Substituting (\ref{previous_switch}) in (\ref{switch1}) and re-applying (\ref{switch_mu}), we get
    \begin{align}\label{switch2}
         V_{\sigma(t)}(t) &\leq  e^{-\Gamma_l(t) } \mu e^{-( \alpha_{l-1} + \beta (t_l - t_{l-1}))} V_{\sigma(t_{l-1})}(t_{l-1}),  \notag \\ 
        &\leq   e^{-\Gamma_l(t) } \mu^2 e^{-( \alpha_{l-1} + \beta (t_l - t_{l-1}))} V_{\sigma(t_{l-1}^-)}(t_{l-1}) \notag \\
        &\leq   e^{-\Gamma_l(t) } \mu^2 e^{-( \alpha_{l-1} + \alpha_{l-2} + \beta (t_l - t_{l-2}))} V_{\sigma(t_{l-2})}(t_{l-2}) \ldots \notag \\
        &\leq  e^{-\Gamma_l(t) } \mu^{N_\sigma(0, t)} e^{-(\alpha_{l-1} + \alpha_{l-2} + \ldots + \alpha_{1} + \alpha_0) } e^{-\beta t_l} V_{\sigma(0)}(0).
    \end{align} 
    Furthermore, we set
    \begin{equation}\label{all alpha}
        \alpha = \min \limits_{k\geq0} \{ {\alpha_k}\}.
    \end{equation}
    Substituting (\ref{all alpha}) in (\ref{switch2}), we arrive at
    \begin{align}\label{switch3}
         V_{\sigma(t)}(t) &\leq  e^{-\Gamma_l(t) } \mu^{N_\sigma(0, t)} e^{-{N_\sigma(0, t)}\alpha} e^{-\beta t_l} V_{\sigma(0)}(0).
    \end{align} 
    Thus, from (\ref{gamma}) and (\ref{switch3}) we get
    \begin{align}\label{switch4}
             V_{\sigma(t)}(t) &\leq  e^{-\int_{t_l}^t a_{\sigma(t),l}(u)du}\mu^{N_\sigma(0, t)} e^{-{N_\sigma(0, t)}\alpha} e^{-\beta t_l} V_{\sigma(0)}(0), \notag \\ & \quad  t_l \leq t. 
    \end{align}
     From (\ref{dwell_definition}), we can conclude that ${N_\sigma(0, t)}$ satisfies
    \begin{equation}\label{switch_num}
       \frac{t}{\bar{\tau}_d} \leq N_\sigma(0, t) \leq \frac{t}{\tau_d},
    \end{equation}
    and since $\alpha > 0$, we apply (\ref{switch_num}) in (\ref{switch4}), which gives via (\ref{inequality for a=b}), (\ref{inequality for a<b})
    \begin{align}\label{trans_stability_result}
             V_{\sigma(t)}(t) &\leq e^{-\int_{t_l}^t a_{\sigma(t),l}(u)du-\beta t_l + \frac{\ln \mu}{\tau_d}t}e^{-{N_\sigma(0, t)}\alpha} V_{\sigma(0)}(0) \notag  \\ 
                              &\leq e^{-\int_{t_l}^t \beta du-\beta t_l + \frac{\ln \mu}{\tau_d}t} e^{-{N_\sigma(0, t)}\alpha}V_{\sigma(0)}(0) \notag \\
                              &\leq e^{-(\beta- \frac{\ln \mu}{\tau_d} + \frac{\alpha}{\bar{\tau}_d})t} V_{\sigma(0)}(0), \quad  t_l \leq t.
    \end{align}
    We notice that inequality (\ref{trans_stability_result}) can hold for any $t \geq 0$. Hence, combining (\ref{kappas})--(\ref{k2}) with (\ref{trans_stability_result}) we reach (\ref{trans_stability_result1}), thus, the target switched system (\ref{Xd_trans}), (\ref{W_trans}) is uniformly exponentially stable for any switching signal $\sigma$ satisfying (\ref{dwell_definition}) with $\tau_d > \tau_d^\star$, for $\tau_d^\star$ as in the statement of the lemma.    
    \end{proof}
    \end{lemma}

   \begin{proof}[Proof of Theorem \ref{main_theorem}]
        Now we are able to complete the proof of Theorem \ref{main_theorem}, and hence, conclude stability of the original closed-loop system (\ref{system}) with (\ref{U1}). Combining (\ref{ut_bound}), (\ref{wt_bound}), with (\ref{trans_stability_result1}), we get (\ref{tra}) where    
        \begin{align}\label{rhoxi}
            \rho &= \sqrt{2\mu\nu_1\nu_2}, \quad 
            \xi  = \frac{\beta- \frac{\ln \mu}{\tau_d} + \frac{\alpha}{\bar{\tau}_d} }{2}. 
        \end{align}
    \end{proof}

  \subsection{Technical Differences with the Stability Strategy in \cite{my_ECC}, \cite{my_TDS}}
    {The stability analysis strategy here employs multiple Lyapunov functionals, thus relaxing the assumption of availability of a common Lyapunov function for the nominal, delay-free, closed-loop systems from \cite{my_ECC}, \cite{my_TDS}. In detail, our {stability proofs are new} and rely on different {backstepping transformations}, than those introduced in \cite{my_ECC}, \cite{my_TDS}. In {Lemma~\ref{lemma_u(theta)}}, the bounds of the estimates regarding norm equivalency are improved compared with the bounds derived in \cite{my_ECC}, \cite{my_TDS}. In {Lemma~\ref{lemma_w(t)}}, the analysis leads to a new $\lambda$ function that explicitly quantifies its dependence on the dwell time (\ref{error lambda})--(\ref{error delta2}), and thus, explicitly quantifies the improvement in, e.g., the decay rate (via the $\alpha$ term in (\ref{trans_stability_result1}) defined in (\ref{alpa}), (\ref{all alpha})).  
    For example, from equations (\ref{error lambda})--(\ref{error delta2}), it follows that the use of dwell time knowledge (resulting in positive $\tau(t)$) directly reduces \(\lambda\), and hence, the upper bound of the prediction mismatch (with respect to the inapplicable, exact predictor-feedback law) presented in (\ref{Wt_bound}) also reduces. This, in turn, results in an improved decay rate of the closed-loop solutions, which can be seen by, for example, observing that $\alpha$ in the exponent of estimate (\ref{trans_stability_result1}) is positive when a minimum dwell time information is employed (whereas $\alpha \equiv 0$ when no dwell time information is used as in \cite{my_ECC}, \cite{my_TDS}). Moreover, Lemma 6 (stability of target system) is now proved by means of a construction of multiple Lyapunov functionals (in contrast with the approaches of \cite{my_ECC}, \cite{my_TDS}, which employ a common Lyapunov functional).  }

\subsection{Proof of Theorem \ref{main_theorem2}}
    We only sketch the proof as it follows similar steps as the ones we employ in the proof of Theorem \ref{main_theorem}. We employ a proof strategy analogous to that of Theorem~\ref{main_theorem}, utilizing the same backstepping transformation, its inverse, and the respective norm equivalency exactly as in Lemmas~\ref{lemma exact predic}--\ref{lemma_u(theta)}. {Next, following the same arguments as in  Lemma~\ref{lemma_w(t)}, we establish a bound for $|W(t)|$ analogous to (\ref{Wt_bound}) as     
    \begin{equation}\label{Wt_bound22}
        | W(t) | \leq \hat{\lambda}(\bar{\epsilon},\tau(t)) \left(  |X(t)| +  \int_{t-D}^{t}{|U(\theta)|d\theta} \right), \quad t \geq 0,
    \end{equation}
    where $\hat{\lambda} : \mathbb{R}_{+} \times [0,\tau_d] \to \mathbb{R}_{+}$ and for each fixed \(\tau \), the map \(\bar{\epsilon} \mapsto \hat{\lambda}(\bar{\epsilon},\tau)\) is a class \(K_{\infty}\) function of \(\bar{\epsilon}\) defined in (\ref{eps2}) and         
    \begin{equation}\label{error lambdaHAT}
            \hat{\lambda}\left(\bar{\epsilon},\tau(t)\right)= \bar{\epsilon}   e^{(M_A + \bar{\epsilon})D} \max \left\{\hat\delta_1(\tau(t)),\hat\delta_2(\bar{\epsilon},\tau(t))\right\},
        \end{equation}
        with
        \begin{align}
            \hat{\delta}_1(\tau(t)) &= \max\left\{1, M_B\right\}  \left[M_K \bigl(D-\tau(t)\bigr) + 1   \right],  \label{error delta1HAT} \\
            \hat{\delta}_2(\bar{\epsilon},\tau(t)) &=   \left[2M_K M_B \bigl(D-\tau(t)\bigr) + \bar{\epsilon} + M_K + M_B  \right], \label{error delta2HAT}
        \end{align}
        for $M_A, \, M_B, \, M_K$ defined in (\ref{M_A}).}
        The main difference is that we now define $\bar\epsilon$ as in (\ref{eps2}) capturing the maximum pairwise distance among all system matrices $(A_i,B_i)$ and all controller gains $K_i$. This is required to make small $\hat{\lambda}$, as this now could be upper bounded by terms that depend on $\bar{\epsilon}$ (instead of $\epsilon$). This difference originates in that controller (\ref{U2}) incorporates the averaged sum of predictors $\hat{P}_{i}(t)$ instead of a single predictor that employs the expected matrix pair $(\bar{A},\bar{B})$. Once these norm bounds and mismatch estimates are derived, we repeat the multiple Lyapunov functional analysis from Lemma~\ref{trans_stability}. By requiring $\bar{\epsilon}$ to be sufficiently small and the minimum dwell time $\tau_d$ sufficiently large, we ensure exponential stability of all target subsystems and, consequently, of the overall switched target system. 
        
        {In particular, we estimate $\bar\epsilon^\star$ as
        \begin{align}\label{eps_condition_22}
        \bar{\epsilon}^\star &= \min \left\{ \lambda_1^{-1}\left(\frac{1}{ \sqrt{2 e^D D \nu_1}}\right), \right. \notag \\ 
        & \left. \quad \lambda_1^{-1}\left(\frac{\min \limits_{i=1,\ldots,p} \left\{ \frac{{\lambda}_{\min}(Q_i)}{|B_i {S_i}|} \right\}}{ \sqrt{2 e^D D (\nu_1+1})} \right)  \right\},
    \end{align}  
    where $\lambda_1^{-1}$ denotes the inverse function of $\lambda_1\left(\bar{\epsilon}\right)=\hat{\lambda}\left(\bar{\epsilon},0\right)$ and similarly for $\bar\tau_d^\star$ we get 
    \begin{equation}\label{dwell_conditionnHAT}
        \bar{\tau}_d^\star= \frac{\ln \mu}{\bar{\beta}},
    \end{equation} 
    where         
    \begin{align}
            \bar{\beta} &= \min \limits_{i=1,\ldots,p} \{\bar{\beta}_i\}, \label{bitaHAT}
    \end{align}
    with
    \begin{align}\label{alpha1HAT}
            \bar{\beta}_i &= \min \left \{ 1-2e^D \hat{\lambda}^2(\bar\epsilon,0) D \nu_1, \notag \right. \\ & \left. \quad \frac{  \frac{1}{2}{\lambda}_{\min}(Q_{i}) - 2 b_{i}e^D \hat{\lambda}^2(\bar{\epsilon},0)(D \nu_1 + 1) }{{\lambda}_{\max}({S_i})} \right \},
    \end{align}
    for $\nu_1$ given in (\ref{nuu1}) and $\mu$ given in (\ref{const_mu}).} 
    Finally, the invertibility of the backstepping transformation implies exponential stability in the original variables $(X,U)$. In particular, since the matrices $(A_j,B_j,K_j)$ play the same role as $(\bar{A},\bar{B},\bar{K})$ in Lemma~\ref{lemma_w(t)}, and thus, the corresponding function $\hat{\lambda}(\bar{\epsilon}, \tau(t))$ remains structurally similar to (\ref{error lambda}), { the exponential stability estimate (\ref{tra2}) holds with a decay rate $\bar\xi$ (similar to $\xi$ defined in (\ref{rhoxi}))     
    \begin{align}\label{rhoxiHAT}
           \bar{\xi}  &= \frac{\bar{\beta}- \frac{\ln \mu}{\tau_d} + \frac{\bar{\alpha}}{\bar{\tau}_d} }{2},
    \end{align}
    where 
        \begin{equation}\label{all alphaHat}
        \bar{\alpha} = \min \limits_{k\geq0} \{ {\bar{\alpha}_k}\},
    \end{equation} with
} 
{
    \begin{equation}\label{alpaHAT}
            \bar{\alpha}_k =  \min \limits_{i=1,\ldots,p} \left\{ \int_{t_k}^{t_k+\tau_d} {\left( \bar{a}_{i,k}(u)- \bar{\beta} \right) du} \right\}, \ k \in \mathbb{Z}_{\geq0},
    \end{equation}
    and 
    \begin{align}\label{alphaHAT}
        \bar{a}_{i,k}(t) &= \min \left \{ 1-2e^D \hat{\lambda}^2(\bar{\epsilon},\tau(t)) D \nu_1, \right. \notag \\ 
        & \left. \quad \frac{  \frac{1}{2}{\lambda}_{\min}(Q_{i}) - 2 b_{i}e^D \hat{\lambda}^2(\bar{\epsilon},\tau(t))(D \nu_1 + 1) }{{\lambda}_{\max}({S_i})} \right \}.
    \end{align}
    Similarly, as Lemmas~\ref{lemma exact predic}--\ref{lemma_u(theta)} are identical, one gets the same overshoot coefficient $\rho$ given in (\ref{rhoxi}). }

\section{Simulation Results}\label{sec4}

\subsection{Responses Under (\ref{U1}) and (\ref{U2})}
Consider the switched system (\ref{system}) with the unstable subsystems \( A_1 \), \( A_2 \), and \( A_3 \)  as
\begin{equation}\label{ex1}
    A_1 = \begin{bmatrix}
    1 & 1 \\
    1 & 2
\end{bmatrix}, \quad
A_2 = \begin{bmatrix}
    0.97 & 1.15 \\
    1.06 & 2.09
\end{bmatrix},  \quad
A_3 = \begin{bmatrix}
    1.08 & 1.2 \\
    1.14 & 2.13
\end{bmatrix},
\end{equation}
and input matrices \( B_1 \), \( B_2 \), and \( B_3 \) as
\begin{equation}\label{ex1_dyn}
    B_1 = \begin{bmatrix}
    0 \\
    1
\end{bmatrix}, \quad
B_2 = \begin{bmatrix}
    0 \\
    1.05
\end{bmatrix}.
\quad
B_3 = \begin{bmatrix}
    0 \\
    1.1
\end{bmatrix}.
\end{equation}
The purpose of our simulations is to validate the theoretical stability results and compare the performance of controllers (\ref{U1}) and (\ref{U2}). {All simulations are  carried out in {MATLAB R2023b} using the \texttt{CVX} toolbox for the convex optimization 
problems and standard built-in functions\footnote{{The code used is provided here \\ {https://github.com/KatsanikJr/stabilization-long-input-delay}.}}.} Hence, for both controllers (\ref{U1}), (\ref{U2}) we choose {by pole placement} 
\begin{align}\label{ex1_k} 
K_1 &= \begin{bmatrix} -13 & -8 \end{bmatrix}, \notag \\ \quad K_2 &= \begin{bmatrix} -10.7742 & -7.6762 \end{bmatrix},  \notag \\ K_3 &= \begin{bmatrix} -10.5564 & -7.4636 \end{bmatrix},
\end{align}
which places the closed-loop poles at $-2,-3$ for each $i=1,2,3$.  { Then, for each subsystem $i \in \{1,2,3\}$, we choose randomly, symmetric positive--definite matrices $Q_i$ and compute the corresponding Lyapunov matrices $S_i$ from (\ref{mean_delayf_stabilityy}) (see also Chapter~4 of \cite{Antsaklis_Michel}). In particular,  we use}
\begin{equation}
    Q_1 = \begin{bmatrix} 1 & 0 \\ 0 & 1 \end{bmatrix}, 
    \qquad
    Q_2 = \begin{bmatrix} 3 & 0 \\ 0 & 3 \end{bmatrix}, 
    \qquad
    Q_3 = \begin{bmatrix} 2 & 0 \\ 0 & 2 \end{bmatrix},
\end{equation}
{and obtain the corresponding matrices $S_i > 0$ as the unique solutions of the Lyapunov equations (\ref{mean_delayf_stabilityy}) as }
\begin{align}
S_1 &= \begin{bmatrix} 3.1 & 0.3 \\ 0.3 & 0.1333 \end{bmatrix}, 
\qquad
S_2 = \begin{bmatrix} 9.168 & 1.0113 \\ 1.0113 & 0.4255 \end{bmatrix}, 
\qquad \notag \\ 
S_3 &= \begin{bmatrix} 6.0907 & 0.6302 \\ 0.6302 & 0.2742 \end{bmatrix}.
\end{align}
Regarding controller (\ref{U1}), we choose the average matrices with the convex optimization routine (\ref{opt_routine}) with respect to the $|\cdot|_2$-norm. {The solution leads to the following optimal average matrices}
\begin{align}\label{exavg}
    \bar{A} &= \begin{bmatrix}
    1.0280 & 1.1077 \\
    1.0837 & 2.0562
\end{bmatrix}, \quad
\bar{B} = \begin{bmatrix}
    0 & 1.05
\end{bmatrix}, \notag \\
\bar{K} &= \begin{bmatrix}
    -11.7782 & -7.7318 
\end{bmatrix}.
\end{align}
\begin{figure}[ht!]
    \centering
    \includegraphics[width=8.8 cm]{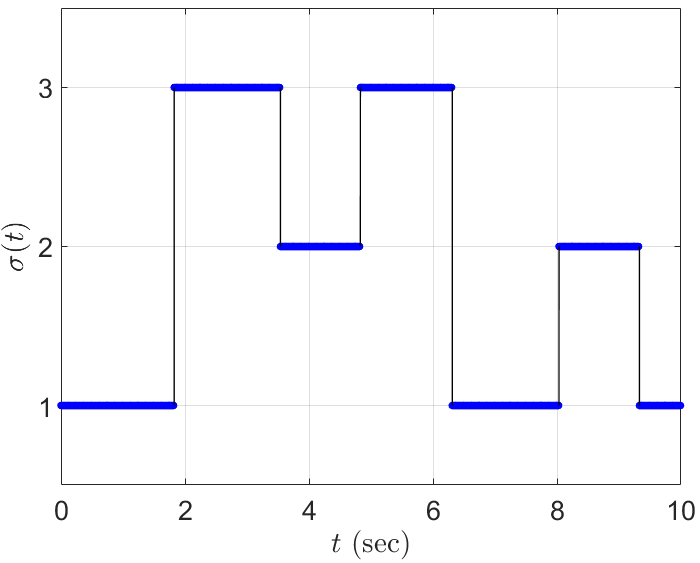}
    \caption{Evolution of switching signal $\sigma(t)$ for all the case studies.
    \label{fig4}}
\end{figure}
\begin{figure}[ht!]
    \centering
    \includegraphics[width=8.8 cm]{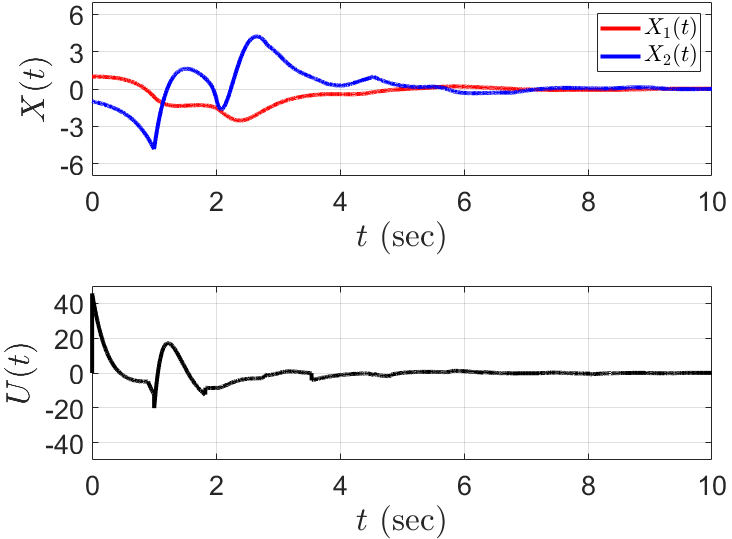}
    \caption{Evolution of state $X(t)$ and control input $U(t)=U_1(t)$ for system (\ref{system}) with (\ref{ex1}), (\ref{ex1_dyn}), under controller (\ref{U1}) with (\ref{exavg}). 
    \label{fig5}}
    \end{figure}    
\begin{figure}[ht!]
    \centering
    \includegraphics[width=8.8 cm]{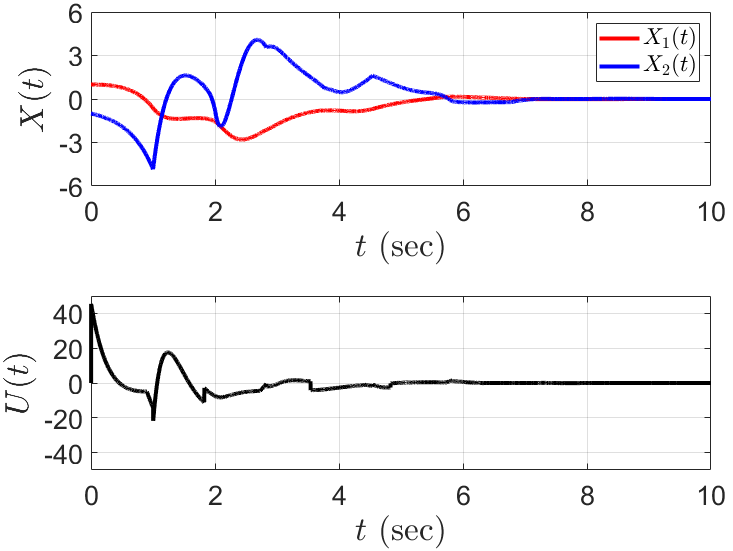}
    \caption{Evolution of state $X(t)$ and control input $U(t)=U_2(t)$ for system (\ref{system}) with (\ref{ex1}), (\ref{ex1_dyn}), under controller (\ref{U2}) with (\ref{ex1_k}). 
    \label{fig6}}
    \end{figure}
We set $D=1$, $\tau_d=0.9$, $\bar{\tau}_d=3$, and initial conditions are $X_0 = \begin{bmatrix} 1 & -1 \end{bmatrix}^{T}$, $U(s)=0$, for $s \in [-D,0)$. Figure~\ref{fig4} shows the evolution of the switching signal over the simulation horizon. Active mode is indicated by a blue line and a black line indicates the switching instants. The system's responses under (\ref{U1}) and (\ref{U2})  are depicted in Figure~\ref{fig5} and Figure~\ref{fig6}, respectively. We observe that the respective, closed-loop performances are similar. {In this example, $\epsilon^\star=\bar{\epsilon}^\star=1.9687\times10^{-11}$, $\tau_d^\star=\bar{\tau}_d^\star=21.002$, whereas $\epsilon=1.2509$ (and $\bar{\epsilon}=2.5018$), which is larger than the theoretical upper bound $\epsilon^\star$. This shows that the designs developed achieve stabilization despite $\epsilon$ (and $\bar{\epsilon}$) and $\tau_d$ exceeding their theoretical bounds, which confirms that the actual $\epsilon^\star$ (resp. $\tau_d^\star$) may be larger (resp. smaller) than the estimated, theoretical values we provide. This is attributed to the conservatism of our stability analysis strategy, for example, the choice of Lyapunov functional (\ref{LyapunovFnc}).}
We also note that closed-loop performance under controller~(\ref{U1}) depends on the specific choice of the optimization routine (\ref{opt_routine}) or, more generally, the selection of the expected matrices ($\bar{A}$, $\bar{B}$, $\bar{K}$) employed in the predictor-based control design. In practice, different choices for expected matrices, such as, e.g., employing an element-wise mean, may lead to different system performance. 
\begin{figure}[ht!]
    \centering
    \includegraphics[width=8.8 cm]{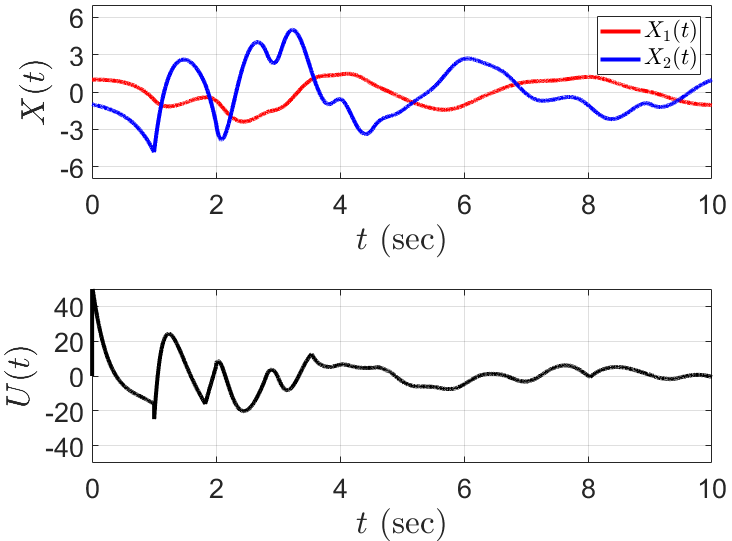}
    \caption{Evolution of state $X(t)$ and control input $U(t)$ for system (\ref{system}) with (\ref{ex1}), (\ref{ex1_dyn}), under controller (\ref{U1}) with (\ref{exavg}), when $\tau(t)=0$ for all $t$. 
    \label{fig7}}
    \end{figure}  
\begin{figure}[ht!]
    \centering
    \includegraphics[width=8.8 cm]{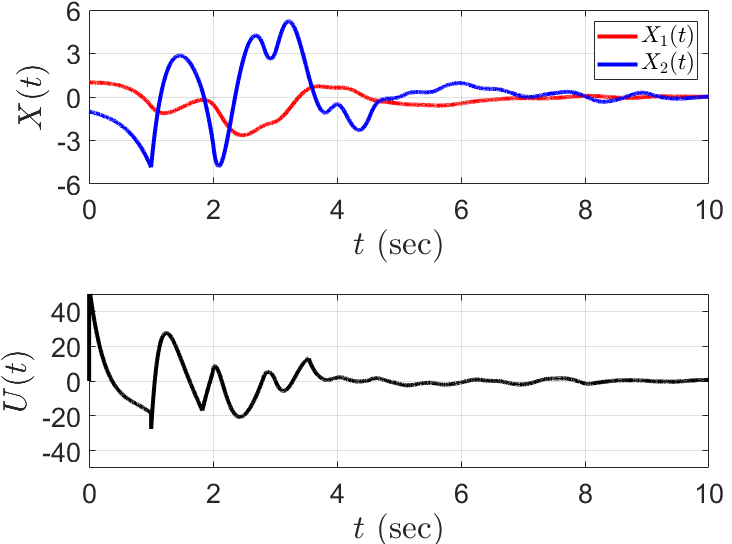}
    \caption{Evolution of state $X(t)$ and control input $U(t)$ for system (\ref{system}) with (\ref{ex1}), (\ref{ex1_dyn}), under controller (\ref{U2}) with (\ref{ex1_k}), when $\tau(t)=0$ for all $t$. 
    \label{fig8}}
    \end{figure}
    \subsection{Performance Comparison without Dwell Time Knowledge and with Exact Predictor Feedback}
    Next, we compare the closed-loop system's behavior, when the controllers are lacking any knowledge regarding the switching signal, {i.e., $\tau(t) = 0, \, \forall t \geq 0,$ which corresponds to our previous predictor-based designs presented in \cite{my_ECC}, \cite{my_TDS}, where no dwell time information is utilized.} As illustrated further in Figures~\ref{fig7} and ~\ref{fig8}, for unknown switching, the performance deteriorates. This is reasonably expected as with no information about the switching signal available, the average predictor-feedback laws employ a less accurate prediction of the state. It should be noted that, in this case, a predictor-based controller, which assumes that the system operates always in a single mode, may fail to even stabilize the system. This observation further reinforces the necessity of employing an average, predictor-based controller.  
    
    {
    Furthermore, we quantify the performance improvement of the designs (\ref{U1}), (\ref{U2}) as compared to the ones presented in \cite{my_ECC}, \cite{my_TDS}, and also compare the performance to the inapplicable, exact predictor-feedback law 
        \begin{equation} \label{Uexact}
            U_{\rm exact}(t)=K_{\sigma(t+D)}P(t),
        \end{equation} where $P$ is given in (\ref{P1(t)}). Specifically, we include the performance index
        \begin{equation}\label{perf_index}
                      J = \int_{0}^{T} \left( | X(t) |^{2} +  | U(t)|^{2} \right) dt,
        \end{equation}
        which quantifies the tracking error and control effort over the simulation horizon. In particular, Table \ref{perf_indexTable} presents the resulting performance indices under both designs (\ref{U1}), (\ref{U2}) for the cases in which there is or not a minimum dwell time available and under (\ref{Uexact}). We observe that the performance improves and is closer to the one under the exact predictor-feedback law according to the performance index, which is also evident by inspection of Figures \ref{fig5}--\ref{fig8}, in which we can observe a slower convergence rate and a more oscillatory behavior in the cases in which no knowledge of the minimum dwell time is employed. This performance improvement can be explained by the fact that utilization of dwell time information enhances the accuracy of construction of the exact future state since an exact prediction is obtained over the known prediction horizon sub-interval $[t, \tau_0(t)+\tau_d]$ (as compared to employing only average estimates over the whole prediction horizon interval $[t, t +D]$). Furthermore, the respective closed-loop response under (\ref{Uexact}) is shown in Figure~\ref{fig:exact_performance}, using the same switching signal and simulation setup. }
        
        \subsection{Responses with Delay Mismatch}
        {
        Additionally, we simulate the closed-loop response under (\ref{U1}) and (\ref{U2}) for perturbations of $\pm 5\%$ around the nominal delay $D$. Results in Figure \ref{perturbations_figs} confirm that the system remains exponentially stable demonstrating that the predictor-based law maintains its stability properties to small delay uncertainty. Based on our simulations, the maximum delay value perturbation that leads to instability or to very poor performance is about $\pm 7\%$. }

        \subsection{Effect of the Number of Modes on Stability Conditions and Complexity of Control Computation}
        {Although performing numerical tests with a larger number of modes may be computationally demanding, the following provide a detailed discussion on the effect of the number of subsystems on the stability conditions and controller design. Adding more modes generally tightens the stability margin $\epsilon^\star$, i.e., the maximum allowable distance among system matrices and among control gains for which exponential stability is guaranteed. To see this, from equations (\ref{error lambda})--(\ref{error delta2}), we observe that the bound $\lambda$ grows with $\bar{M}_A, \, \bar{M}_B, \, M_K,$ whose values increase with the increase of the number of modes (as they represent the maximum over the norms of the matrices $A_i, \, B_i,\, K_i$, respectively, over all $i$). In turn, this implies that, the  parameter $\epsilon^{\star}$ (see (\ref{eps_condition_2})) decreases with the increase of $\bar{M}_A, \, \bar{M}_B, \, M_K$. Similar reasoning applies due to the dependency of $\epsilon^\star$ on $\nu_1$, $\nu_2$ (whose values increase with $M_A, \, M_B, \, M_K$ according to (\ref{nuu1}), (\ref{nuu2})). }
        
        {
        Furthermore, in the optimization problem (\ref{opt_routine}), (\ref{opt_routine2}), a larger number of modes increases the number of constraints used to compute $(\bar{A}, \bar{B}, \bar{K})$, which increases linearly their computational complexity, while not affecting though the structure of the algorithm. A larger number of modes implies that one has to compute a larger number of control gains $K_i$, which may increase computational effort linearly with the number of modes. However, since the gains $K_i$ are computed once, {a priori}, the {online controller implementation} for the case of $U(t)=U_1(t)$ remains simple. On the other hand, for the case of controller $U(t)=U_2(t)$, since it additionally involves real-time computation of a larger set of (exact) predictor states, the online computational load may increase as the number of modes grows. }

\begin{table}[ht]
    \centering
    \caption{Cost $J$ as in (\ref{perf_index}) under each controller.}
    \label{perf_indexTable}
    \begin{tabular}{p{6cm} c}
        \toprule
        \textbf{Controller $U$} & \textbf{Cost $J$} \\
        \midrule
        $U_1$ feedback law as in (4)  
            & $3.68 \times 10^{2}$ \\[3pt]

        $U_1$ feedback law with $\tau(t) = 0,\; \forall t \ge 0$ (no dwell time knowledge)  
            & $8.34 \times 10^{2}$ \\[3pt]

        $U_2$ feedback law as in (11)   
            & $3.83 \times 10^{2}$ \\[3pt]

        $U_2$ feedback law with $\tau(t) = 0,\; \forall t \ge 0$ (no dwell time knowledge)  
            & $7.81 \times 10^{2}$ \\[3pt]

        $U_{\mathrm{exact}}$ feedback law as in (\ref{Uexact}) (with exact predictor)  
            & $2.32 \times 10^{2}$ \\
        \bottomrule
    \end{tabular}
\end{table}

        \begin{figure}[ht!]
            \centering
            \includegraphics[width=8.8 cm]{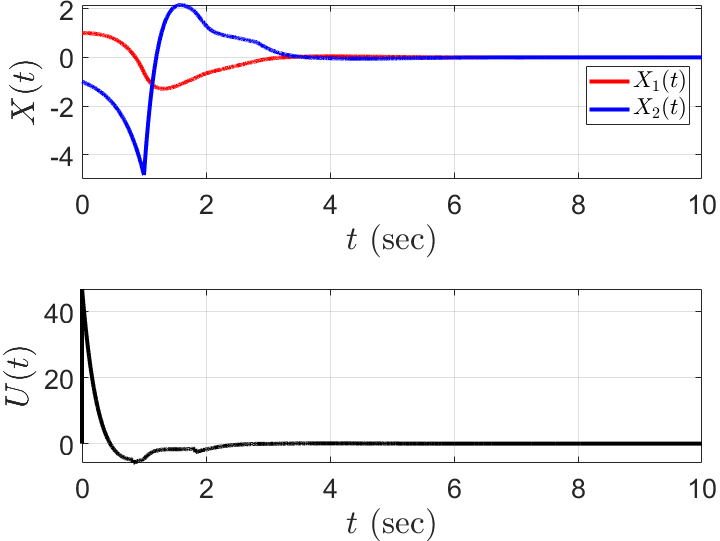}
            \caption{Closed-loop response under the inapplicable, exact predictor-feedback law $U_{\rm exact}$ in (\ref{Uexact}) with (\ref{ex1_k}).}
            \label{fig:exact_performance}
        \end{figure}

        \begin{figure}[ht!]
            \centering
            \begin{subfigure}{0.49\linewidth}
                \centering
                \includegraphics[width=\linewidth]{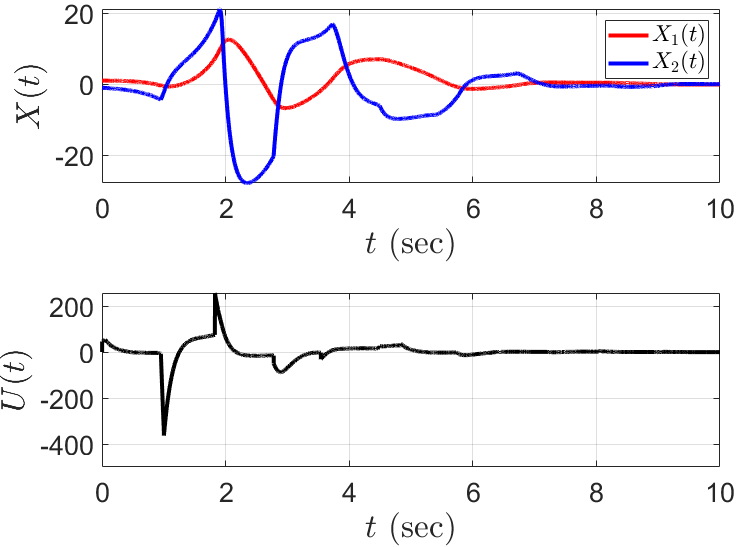} 
                \caption{$U_1$: actual $D=0.95$}
                \label{fig:robust_U1_m5}
            \end{subfigure}\hfill          
            \begin{subfigure}{0.49\linewidth}
                \centering
                \includegraphics[width=\linewidth]{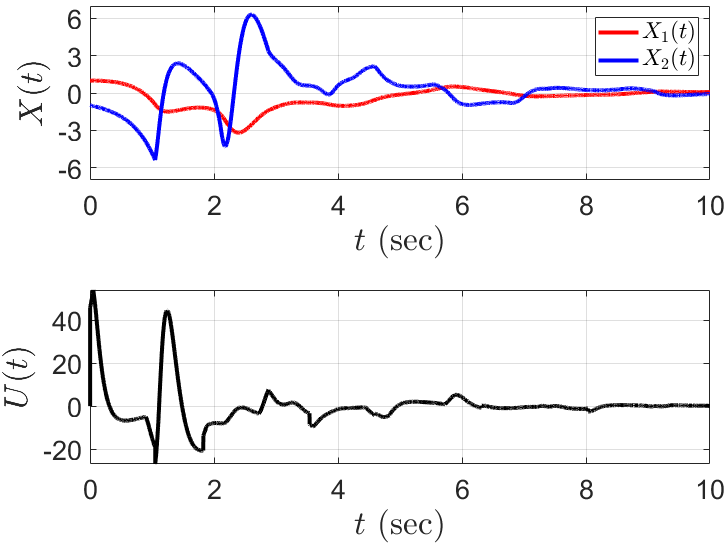} 
                \caption{$U_1$:  actual $D=1.05$}
            \label{fig:robust_U1_p5}
            \end{subfigure}
            \vspace{0.75em}
            \begin{subfigure}{0.49\linewidth}
                \centering
                \includegraphics[width=\linewidth]{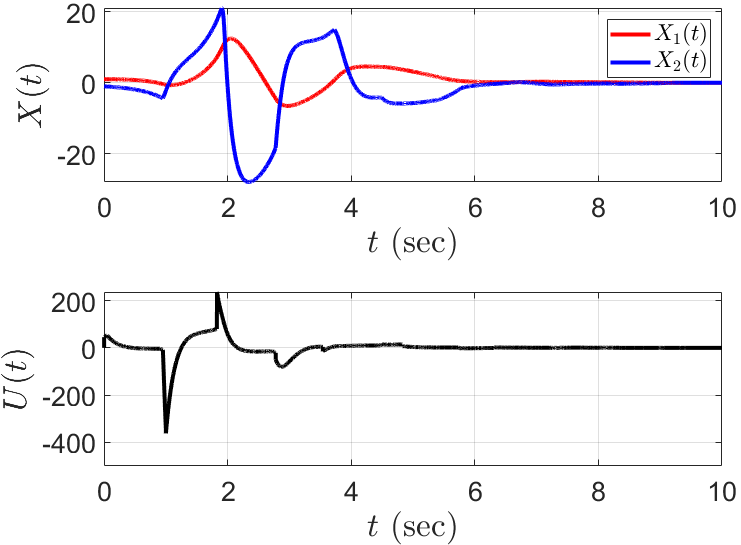} 
                \caption{$U_2$: actual $D=0.95$}
                \label{fig:robust_U2_m5}
            \end{subfigure}\hfill            
            \begin{subfigure}{0.49\linewidth}
                \centering
                \includegraphics[width=\linewidth]{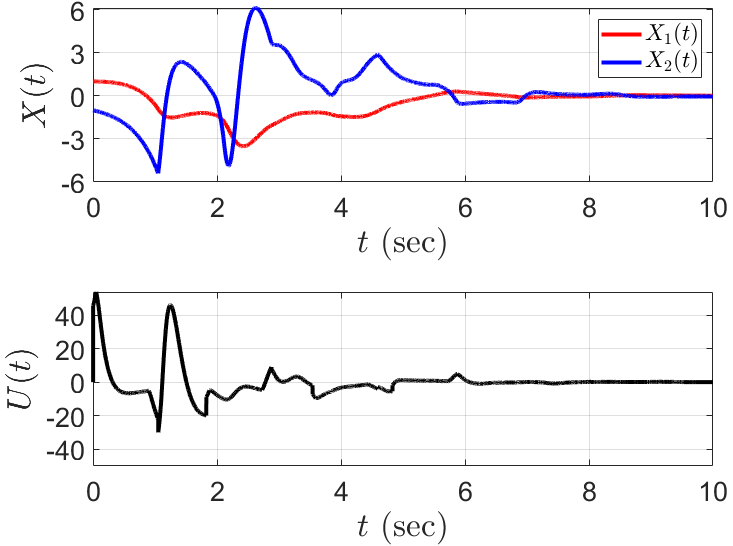} 
                \caption{$U_2$: actual $D=1.05$}
                \label{fig:robust_U2_p5}
            \end{subfigure}
            \vspace{0.75em}
        \caption{Robustness to delay mismatch. Closed-loop responses with $\pm 5\%$ perturbation of the nominal delay $D$, under $U_1$ and $U_2$.}
        \label{perturbations_figs}
    \end{figure}

\section{Conclusions and Future Work}\label{sec5} 
In this work, we developed two different, delay compensating control designs for switched linear systems with long input delay, where the switching signal is time-dependent, featuring a minimum/maximum dwell time. The first control design employs an average predictor-feedback law that properly exploits minimum dwell-time information, allowing exact short-term prediction of the state. The second control law is based on averaging exact predictor-feedback laws for each subsystem. Uniform exponential stability, for the respective closed-loop systems, is proved via introduction of backstepping transformations and construction of multiple Lyapunov functionals. Numerical simulations confirm the effectiveness of the proposed control laws, also providing comparison with cases when the controllers lack any knowledge about the switching signal. 

Summarizing our results we can conclude that, while controller (\ref{U2}) offers greater flexibility in terms of tuning—since each subsystem's (exact) predictor feedback law can be independently adjusted—controller (\ref{U1}) is preferable for theoretical analysis, as it corresponds to the exact predictor-feedback law for an expected system. Such a correspondence could facilitate a potential stability analysis and guide the controller's tuning. With respect to state prediction accuracy, both controllers behave similarly overall, but such comparisons remain scenario-dependent and further theoretical/numerical comparisons are needed, for potentially exactly quantifying the state prediction error for each controller.

As future work, we aim at developing design approaches that exploit additional, available dwell time and other, switching signal information, as, e.g., in the case of stochastic switching (see, e.g., \cite{Kong_P1}), towards construction of accurate predictions of the state with larger allowable differences among system's matrices and among controller's gains. {We also aim at addressing systems featuring impulse effects (in addition to input delay) since such systems may appear in certain applications (depending on the dynamic models and respective states considered). Allowing the state to have discontinuities (i.e., considering an impulsive, rather than a switched, system) would require adoption of a different, hybrid systems framework, which would in turn require construction of a new predictor-based control law and introduction of a different Lyapunov-based stability analysis, within a hybrid systems setup.}

\bibliography{elsarticle}      
\bibliographystyle{elsarticle-num}

\end{document}